\documentclass[12pt]{article}
\title{\vspace{-2 cm} Magnetic ring chains with vertex coupling \\ of a preferred orientation}

\usepackage{graphicx}
\usepackage{amssymb}
\usepackage{commath}
\usepackage{caption}
\usepackage[textwidth = 7in]{geometry}
\usepackage{mathtools}
\usepackage[dvipsnames]{xcolor}
\usepackage{enumitem}
\usepackage{amsmath}
\usepackage{amsthm}
\usepackage{thmtools}
\usepackage{hyperref, cleveref}
\usepackage{setspace}

\newcommand{\e}{\mathrm{e}}
\newcommand{\D}{\mathrm{d}}

\stepcounter{secnumdepth}
\stepcounter{tocdepth}

\DeclareMathOperator\arctanh{arctanh}

\pagebreak

\newtheorem{thm}{Theorem}

\pagebreak

\author{Marzieh Baradaran$^1$, Pavel Exner$^{2,3}$, and Ji\v{r}\'{\i} Lipovsk\'{y}$^1$}

\date{\small 1) Department of Physics, Faculty of Science, University of Hradec Kr\'alov\'{e}, Rokitansk\'eho 62, 500 03 Hradec Kr\'alov\'{e}, Czechia \\
2) Doppler Institute for Mathematical Physics and Applied Mathematics, Czech Technical University, B\v rehov\'a 7, 11519 Prague, Czechia \\
3) Department of Theoretical Physics, Nuclear Physics Institute, Czech Academy of Sciences, 25068 \v{R}e\v{z} near Prague, Czechia \\
\emph{marzie.baradaran@yahoo.com, exner@ujf.cas.cz, jiri.lipovsky@uhk.cz}}

\begin{document}

\captionsetup[figure]{labelfont={bf},labelformat={default},labelsep=period,name={Fig.}}
\maketitle

\begin{abstract}
We discuss spectral properties of a periodic quantum graph consisting of an array of rings coupled either tightly or loosely through connecting links, assuming that the vertex coupling is manifestly non-invariant with respect to the time reversal and a homogeneous magnetic field perpendicular to the graph plane is present. It is shown that the vertex parity determines the spectral behavior at high energies and the Band-Berkolaiko universality holds whenever the edges are incommensurate. The magnetic field influences the probability that an energy belongs to the spectrum in the tight-chain case, and also it can turn some spectral bands into infinitely degenerate eigenvalues.
\end{abstract}

\section{Introduction}
\label{sect:Intro}

The class of models usually called `quantum graphs' offers both many interesting mathematical problems and a number of applications in physics; for a survey and bibliography see \cite{BK13, EKKST}. To define a self-adjoint Laplacian, or more generally a Schr\"odinger operator, on a metric graph one has to specify the conditions matching the functions of its domain in the graph vertices. This can be done in many different ways \cite[Thm.~1.4.4]{BK13}, but in most cases the simple condition known as the $\delta$-coupling is used, the only one consistent with the continuity at the vertices, frequently its simplest version called Kirchhoff (alternatively free, standard, etc.). It is true that every self-adjoint coupling can be given meaning in terms of scaled families of magnetic Schr\"odinger operator on Neumann network built over the graph \cite{EP13}, but this result has primarily an existence meaning, in applications one rather chooses the vertex matching conditions \emph{ad hoc} to fit the nature of the physical problem they have to describe.

A class of interest here are conditions giving rise to a dynamics violating the time-reversal invariance. A motivation to investigate such couplings came from an attempt to use quantum graphs to model the anomalous Hall effect
\cite{SK15} through lattices with the $\delta$-coupling at the nodes. To achieve the goal, the model required a restriction on the family of acceptable states, hard to justify from the first principles. To get a preferred orientation of the motion at the vertices, it is more natural to choose an appropriate vertex coupling, different from the $\delta$. A simple example was proposed in \cite{ET18} and studied further in \cite{BE21, BET20, BET21, EL19}; it was noted that such a model has interesting property, namely that its spectral and scattering property at high energies depend on the parity of the vertices involved. Moreover, it was made clear recently \cite{ET21} that the root of this behavior is in the presence/absence of the Dirichlet component of the vertex coupling \cite[Thm.~1.4.4]{BK13}, and that while such graphs violate the time-reversal invariance, they preserve the $\mathcal{PT}$-symmetry.

The aim of the present paper is to revisit the topic of \cite{BET20, BET21} and to investigate what happens if the system is placed into a homogenous magnetic field. The interest to this problem comes from the fact that we have here two different effects violating the time-reversal invariance which can cooperate or compete. Specifically, we consider a periodic chain graph consisting of an infinite array of rings connected either tightly, or loosely through connecting links, assuming the coupling condition proposed in \cite{ET18} at the vertices. We suppose that the particle, the motion of which is confined to such a graph, is charged and exposed to a magnetic field of intensity $B$ perpendicular to the graph plane; the corresponding vector potential $\mathbf{A}$ can be chosen to be tangent to each ring and constant there with the modulus $A=\tfrac{1}{2\pi}\Phi$, where $\Phi$ is the magnetic flux through the loop (without loss of generality we may suppose that the rings are of unit radius).

The spectrum of such chains has a band-gap character, and as expected, its properties depend strongly on the vertex parity. In the generic `loose' chain with vertices of degree three the spectrum is dominated by gaps at high energies. There are two degenerate cases where the vertices are of degree four and the bands cover a more significant part of the spectrum. A quantity characterizing their role is the probability that a randomly chosen positive number belongs to the spectrum, defined by \eqref{probsigma} below, which is nonzero in the degenerate case. What is important, we prove that Band-Berkolaiko universality \cite{BB13} holds in this case, that is, the said probability value is independent of the graph edge lengths as long as they are incommensurate. On the other hand, we show that in the degenerate case the said probability depends generically on the magnetic field. Another field influence concerns the character of spectral bands, it may turn some of them into flat ones, i.e. infinitely degenerate eigenvalues. The effect is not that dramatic as for a tight chain with the $\delta$-coupling where a field with the flux half-integer in the natural ($2\pi$) units can make the spectrum pure point \cite[Thm.~2.1]{EM17}, but it is present and half-integer values of $A$ play again an important role.

\section{The model}
\label{sect:GenMod}

In the presence of a magnetic field, the Hamiltonian describing a particle of unit charge living on the graph acts as the magnetic Laplacian, $\left(-i\nabla-\mathbf{A}\right)^2$, provided we use the rational system of units, $\hbar=2m=1$. The motion on the graph edges is one-dimensional, which means that we replace the usual Laplacian acting on the $j$th edge as $\psi_{j}\mapsto -\psi_{j}''$ by $\psi_{j}\mapsto -\mathcal{D}^{2}\psi_{j}$, where
\begin{equation}\label{quasi-derivative}
   \mathcal{D}:=\frac{\D}{\D x}- i\, A_{j},
\end{equation}
the particle charge is set to $-1$, and $A_{j}$ is the tangent component of the magnetic vector potential on the $j$th edge. The domain of the operator consists of the functions such that $\mathcal{D}^{2}\psi_{j}\in L^2$ on each edge and at the vertices the functions are matched by conditions ensuring the self-adjointness \cite{KS03}. This is a large family -- for a vertex of degree $N$ it depends on $N^2$ parameters -- out of which we select the one-parameter subfamily \cite{ET18} of couplings violating time-reversal symmetry, namely
\begin{equation}\label{coupB}
(\psi_{j+1}-\psi_{j})+i\ell \left(\mathcal{D}\psi_{j+1}+\mathcal{D}\psi_{j} \right)=0,\quad\; j=1,\dots,N,
\end{equation}
where, with an abuse of notations, we use the symbols $\psi_j$ for the boundary value of the function $\psi_j$ at the vertex, and similarly $\mathcal{D}\psi_{j}$ for the boundary value the function $\mathcal{D}\psi_{j}$  with the
derivative taken in the outward direction; the parameter $\ell$ is a positive number fixing the length scale.

The graph we are interested in is shown in Fig.~\ref{fig1}.
\begin{figure}[h]
\centering
\includegraphics[scale=1]{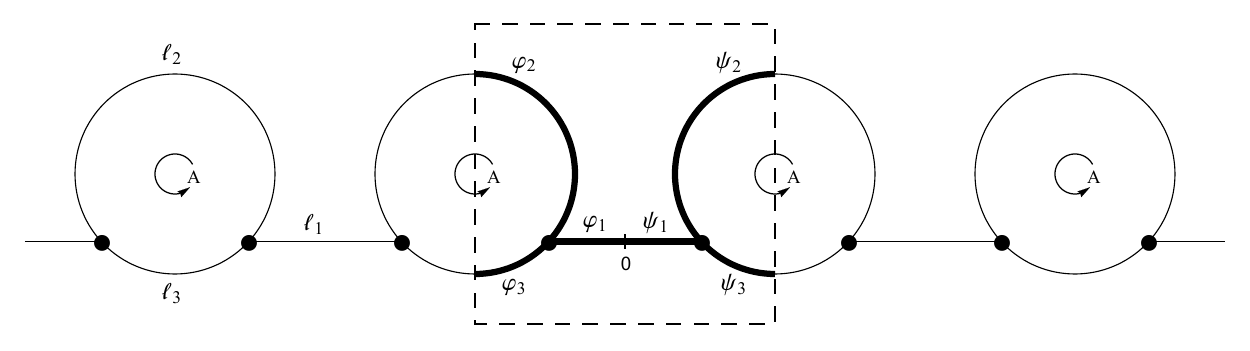}
\caption{ An elementary cell of the ring chain graph}
\label{fig1}
\end{figure}
In general, we suppose that all the lengths involved are nonzero, $\ell_j>0$, $j=1,2,3$, and consequently, the vertices are of degree three; degenerate situations where one of the three edges shrinks to zero will be considered separately in the next sections. Furthermore, the presence of the parameter $\ell$ in the conditions \eqref{coupB} allows us to assume, without loss of generality, that the ring circumference is fixed by assumption being equal, say, to $\ell_{2}+\ell_{3}=2\pi$; results for the other values will follow by a simple scaling transformation. Note also that while we think of the graph loops as of circular rings, their shape does not really matter; the important quantities are the loop perimeter and the magnetic flux threading the loop. We are interested in the periodic situation where the coupling \eqref{coupB} at all the vertices is the same and applied magnetic field is homogeneous. The corresponding vector potential $A$ can be chosen in different ways, for the sake of definiteness we fix the gauge so that the tangent component of $A$  is constant on each graph edge. By Stokes theorem the loop integral of $A$ is the flux, the quantum of which in the natural units equals $2\pi$. Moreover, the freedom of choosing the gauge allows us to put $A_j=0$ on the links connecting the loops so that $\mathcal{D}$ there is just the usual derivative.

By the periodicity hypothesis, the spectral analysis can be performed using the Floquet-Bloch decomposition \cite[Chap.~4]{BK13}. To this aim, we consider an elementary cell of the graph as indicated in Fig.~\ref{fig1}. To parametrize the edges we choose coordinates to increase always `from left to right', in which case the constant components of $A_j$ in \eqref{quasi-derivative} will have opposite signs on the upper and lower arcs. For definiteness again we choose the value $-A$ on the upper arcs and $+A$ on the lower ones; the magnetic Laplacian operator acts then as $-\mathcal{D}_{\pm}^{2}:=-\big( \frac{d}{dx}\pm i A \big)^{2}$ with the upper and lower signs corresponding to the upper and lower arcs, respectively.

The main question we are going to address here concerns the spectrum of the described quantum graph Hamiltonian. The fiber operators in the decomposition have a purely discrete spectrum with the eigenfunctions which at the  energy $k^2$ are linear combination of the functions $\e^{\pm ikx}$, $\e^{-iAx}e^{\pm ikx}$, and $\e^{iAx}e^{\pm ikx}$ at the connecting edges, the upper and lower arcs, respectively. In other words, we employ the following Ansatz,
\begin{align}\label{ansatzB}
&\psi_{1}(x)=a_{1}^{+}\e^{ikx}+a_{1}^{-}\e^{-ikx},                          & x&\in[0,\tfrac{1}{2}\ell_{1}]  , \nonumber \\
&\psi_{2}(x)=\big(a_{2}^{+}\e^{ikx}+a_{2}^{-}\e^{-ikx}\big)\e^{-iAx},        & x&\in[0,\tfrac{1}{2}\ell_{2}]  , \nonumber \\
&\psi_{3}(x)=\big(a_{3}^{+}\e^{ikx}+a_{3}^{-}\e^{-ikx}\big)\e^{iAx},      & x&\in[0,\tfrac{1}{2}\ell_{3}]  ,\nonumber \\[7pt]
&\varphi_{1}(x)=b_{1}^{+}\e^{ikx}+b_{1}^{-}\e^{-ikx},                       & x&\in[-\tfrac{1}{2}\ell_{1},0] , \\
&\varphi_{2}(x)=\big(b_{2}^{+}\e^{ikx}+b_{2}^{-}\e^{-ikx}\big)\e^{-iAx},  & x&\in[-\tfrac{1}{2}\ell_{2},0] , \nonumber \\
&\varphi_{3}(x)=\big(b_{3}^{+}\e^{ikx}+b_{3}^{-}\e^{-ikx}\big)\e^{iAx},   & x&\in[-\tfrac{1}{2}\ell_{3},0] . \nonumber
\end{align}
At the midpoint of the connecting link, the functions $\psi_1$ and $\varphi_1$ have to be matched smoothly; together with the conditions imposed by the Floquet-Bloch decomposition at the loose ends of the cell graph, we have
\begin{align}\label{FloqB}
\psi _2\big(\tfrac{1}{2}\ell_{2}\big)&=\e^{i \theta }\, \varphi _2\big(\!-\!\tfrac{1}{2}\ell_{2}\big),
& \mathcal{D}_{+}\psi _2\big(\tfrac{1}{2}\ell_{2}\big)&=\e^{i \theta } \,\mathcal{D}_{+}\varphi _2\big(\!-\!\tfrac{1}{2}\ell_{2}\big),\nonumber \\
\psi _3\big(\tfrac{1}{2}\ell_{3}\big)&=\e^{i \theta } \,\varphi _3\big(\!-\!\tfrac{1}{2}\ell_{3}\big),
& \mathcal{D}_{-}\psi _3\big(\tfrac{1}{2}\ell_{3}\big)&=\e^{i \theta }\, \mathcal{D}_{-}\varphi _3\big(\!-\!\tfrac{1}{2}\ell_{3}\big),  \\
\psi _1(0)&=\varphi _1(0),       &  \psi _1'(0)&=\varphi _1'(0).  \nonumber
\end{align}
Imposing then the matching conditions \eqref{coupB} at the two vertices of the cell, keeping in mind that the derivatives are taken in the outward direction, we arrive at the following set of equations
\begin{align}\label{conditionB}
& \psi _3(0)-\psi _1\left(\tfrac{1}{2}\ell_{1}\right)  +i \ell  \left(\mathcal{D}_{-}\psi _3(0)-\psi _1'\left(\tfrac{1}{2}\ell_{1}\right)\right)=0,\nonumber \\
& \psi _2(0)-\psi _3(0)  +i \ell  \left(\mathcal{D}_{+}\psi _2(0)+\mathcal{D}_{-}\psi _3(0)\right)=0,\nonumber \\
& \psi _1\left(\tfrac{1}{2}\ell_{1}\right)-\psi _2(0) +i \ell  \left(\mathcal{D}_{+}\psi _2(0)-\psi _1'\left(\tfrac{1}{2}\ell_{1}\right)\right)=0,\\[7pt]
& \varphi _2(0)-\varphi _1\left(-\tfrac{1}{2}\ell_{1}\right) +i \ell  \left(\varphi _1'\left(-\tfrac{1}{2}\ell_{1}\right)-\mathcal{D}_{+}\varphi _2(0)\right)=0,\nonumber \\
& \varphi _3(0)-\varphi _2(0)-i \ell  \left(\mathcal{D}_{+}\varphi _2(0)+\mathcal{D}_{-}\varphi _3(0)\right)=0,\nonumber \\
& \varphi _1\left(-\tfrac{1}{2}\ell_{1}\right)-\varphi _3(0) +i \ell  \left(\varphi _1'\left(-\tfrac{1}{2}\ell_{1}\right)-\mathcal{D}_{-}\varphi _3(0)\right)=0.\nonumber
\end{align}
Substituting from \eqref{ansatzB} into \eqref{FloqB} allows us to eliminate the coefficients $a_j^{\pm}$ expressing them in terms of $b_j^{\pm}$, $j=1,2,3$; taking then \eqref{conditionB} into account, we get a system of six linear equations which is solvable provided its determinant vanishes; this yields the spectral condition
\begin{align}
& - \left(2 k \ell \, \sin  A \ell _2 \, \cos  A \ell _3 +\left(k^2 \ell ^2+1\right) \sin  k \ell _2  \,\cos  k \ell _3 \right) \cos k \ell_1 \\ \nonumber
&+\frac{1}{2} \sin  k \ell _1  \left(\left(k^4 \ell ^4+3\right) \sin  k \ell _2  \,\sin  k \ell _3 -2 \left(k^2 \ell ^2+1\right) \sin  A \ell _2  \,\sin  A \ell _3 \right) \\ \nonumber
&+ \left(\left(k^2 \ell ^2+1\right) \cos  A \ell _3  \,\sin  k \ell _1  -2 k \ell  \sin  A \ell _3  \,\cos  k \ell _1 \right)\cos  A \ell _2 -\left(k^2 \ell ^2+1\right)\cos  k \ell _2  \,\sin k \left(\ell _1+\ell _3\right)   \\ \nonumber
&+\cos \theta  \left(\left(k^2 \ell ^2+1\right) \left(\cos  A \ell _3  \,\sin  k \ell _2 +\cos  A \ell _2  \,\sin  k \ell _3 \right)+2 k \ell  \,\sin A \ell _2  \,\cos  k \ell _3 +2 k \ell  \,\sin  A \ell _3  \,\cos  k \ell _2 \right)  \\ \nonumber
&+\sin \theta  \left(\left(k^2 \ell ^2+1\right) \left(\sin A \ell _3  \,\sin  k \ell _2 -\sin  A \ell _2 \, \sin  k \ell _3 \right)-2 k \ell \, \cos A \ell _3  \,\cos  k \ell _2 +2 k \ell  \cos  A \ell _2  \,\cos  k \ell _3 \right) \\\nonumber
&=0,
\end{align}
which can be rewritten in a simpler form
\small
\begin{align}\label{SC,gen,ell-1,3}
& \,  8 \cos \theta  \left((k \ell +1)^2 \sin(A+k)\pi \;\cos  (A-k)(\pi -\ell _3) -(k \ell -1)^2 \sin (A-k)\pi \; \cos  (A+k)(\pi -\ell _3) \right) \notag \\
+& \, 8 \sin \theta  \left((k \ell -1)^2 \sin (A-k)\pi \;\sin (A+k)(\pi -\ell _3)-(k \ell +1)^2 \sin(A+k)\pi \;\sin (A-k)(\pi -\ell _3)\right)  \notag \\
+& \, (k \ell -1)^2 \left(4 \sin \left(2 \pi  A+k \ell _1\right)+(k \ell +1)^2 \left(\sin  k \left(2 \pi -\ell _1\right) +2 \sin  k \ell _1 \;\cos 2 k \left(\pi -\ell _3\right) \right)\right)    \notag \\
-& \, 4 (k \ell +1)^2 \sin \left(2 \pi  A-k \ell _1\right)-(k^2 \ell ^2+3)^2 \sin k \left(\ell _1+2 \pi \right)=0,
\end{align}
\normalsize
taking into account the assumption $\ell_2=2\pi-\ell_3$. As usual in Bloch analysis, we are looking at eigenvalues of the fiber operators, i.e. solutions of the above condition, and their dependence on the quasimomentum $\theta$. We note that the left-hand side is an analytic function of the variables $k$ and $\theta$ so that by the implicit function theorem \cite[Chap.~VIII]{Ja56} the solutions are also analytic.

With a later purpose in mind, we rewrite the condition \eqref{SC,gen,ell-1,3} as
\begin{equation}\label{formul,abc}
a \cos \theta +b \sin \theta =c,
\end{equation}
where
\small
\begin{align}\label{SC,gen,abc}
a:= \; & 8  \left((k \ell +1)^2 \sin(A+k)\pi \;\cos (A-k)(\pi -\ell _3)-(k \ell -1)^2 \sin (A-k)\pi  \;\cos (A+k)(\pi -\ell _3)\right)   ,\notag\\[4pt] \nonumber
b:= \; & 8  \left((k \ell -1)^2 \sin (A-k)\pi \;\sin (A+k)(\pi -\ell _3)-(k \ell +1)^2 \sin(A+k)\pi \;\sin (A-k)(\pi -\ell _3)\right)  , \notag\\[4pt]
c:= \; & -(k \ell -1)^2 \left(4 \sin \left(2 \pi  A+k \ell _1\right)+(k \ell +1)^2 \left(\sin  k \left(2 \pi -\ell _1\right) +2 \sin  k \ell _1 \; \cos 2 k \left(\pi -\ell _3\right) \right)\right)      \notag\\
\;& + \,    4 (k \ell +1)^2 \sin \left(2 \pi  A-k \ell _1\right)+ (k^2 \ell ^2+3 )^2 \sin k \left(\ell _1+2 \pi \right).
\end{align}
\normalsize
One situation when condition \eqref{formul,abc} is satisfied occurs if all the three functions vanish simultaneously. Such a solution is independent of $\theta$ and corresponds thus to an infinitely degenerate eigenvalue. If $a^2+b^2\neq 0$, we can write $\sin \vartheta=\frac{a}{\sqrt{a^2+b^2}}$ and $\cos \vartheta=\frac{b}{\sqrt{a^2+b^2}}$ and cast the condition \eqref{formul,abc} into the form
$$
\sin (\vartheta+\theta)=\frac{c}{\sqrt{a^2+b^2}}\;.
$$
This shows, in general, that a number $k^2$ belongs to the spectrum if
\begin{equation}\label{bandCon,a2b2c2}
a^2+b^2-c^2\geq 0
\end{equation}
provided $a^2+b^2\neq 0$. In the opposite case when the left-hand side of \eqref{bandCon,a2b2c2} is negative, the number $k^2$ belongs to the gaps. We will refer to \eqref{bandCon,a2b2c2} here and in the following as to the band condition.
\begin{thm}\label{thmGen}
The spectrum has the following properties:
\begin{enumerate}[label=\textnormal{(\roman*)}]
\setlength{\itemsep}{-3pt}
\item \label{th1} For $A\in\mathbb{Z}$, the spectrum coincides with that of the non-magnetic chain analyzed in \cite{BET21}.
\item \label{th2} For  $ A-\frac 12 \in\mathbb{Z} $, the spectrum contains infinitely degenerate eigenvalues equal to $k^2=q^2 \left(n-\frac 12\right)^2$ with $q,n \in\mathbb{N}$ where $q$ is odd.
\item \label{th3a} For $ 2A\notin\mathbb{Z}$, flat bands appear for $A+\ell^{-1}\in\mathbb{Z}$ at the energy $k^2=\ell^{-2}$ independently of the other parameters.
\item \label{th3b} In addition, in the general model with $\ell_1>0$ it may happen that the spectral bands shrinks to points for particular values of parameters. In the asymmetric case, $\ell_3\ne\pi$, this is the case for $k=\frac{m\pi}{2(\pi -\ell_3)}$, $m\in\mathbb{Z}$, together with those values of $\ell_3$ for which $2k\notin\mathbb{N}$, the number $k^2$ then belongs to the spectrum for $A=-\pi^{-1}\,\arctan \big(\frac{2k\ell}{k^2 \ell ^2+1}\,\tan k\pi \big)+m'$ (with odd $m$), and $A=-\pi^{-1}\,\arctan \big(\frac{k^2 \ell ^2+1}{2k\ell}\,\tan k\pi\big)+m'$ (with even $m$), $m'\in\mathbb{Z}$, together with particular values of $\ell_1$; for a given energy $k^2$, such exceptional situations happen only once in the domain $A\in(0,1)$. In the symmetric case for any $2k\notin\mathbb{N}$, the number $k^2$ belongs to the spectrum for $A=-\pi^{-1}\,\arctan \big(\frac{k^2 \ell ^2+1}{2k\ell}\,\tan k\pi\big)+m'$ and particular values of $\ell_1$.
\item \label{th4} Away from those flat bands, the spectrum is absolutely continuous having a band-and-gap structure; it has infinitely many gaps in its positive part.
\item \label{th5} The negative spectrum consists of a pair of bands which may merge at particular values of the parameters.
\end{enumerate}
\end{thm}
\begin{proof}
The first claim follows from a simple gauge transformation, however, one can also check its validity directly. For $A\in\mathbb{Z}$, the functions $a$, $b$, and $c$ in \eqref{SC,gen,abc} simplify to
\begin{align}\label{gen-A-integ,abc}
a= \;& 8  \sin k\pi \left(2 \left(k^2 \ell ^2+1\right) \cos A\ell_3 \,
   \cos k(\pi -\ell_3 )-4 k \ell \, \sin A \ell_3 \,\sin k(\pi -\ell_3)\right),\notag\\ \nonumber
b= \; & 8  \sin k\pi \left(2 \left(k^2 \ell ^2+1\right) \sin A \ell_3 \,
   \cos k(\pi -\ell_3 )+4 k \ell  \cos A \ell_3 \, \sin k(\pi -\ell_3)\right), \notag\\
c= \; & -\left(k^2 \ell ^2-1\right)^2 \left(2 \sin k \ell_1 \, \cos 2k(\pi -\ell_3 )+\sin k(2
   \pi-\ell_1) \right)   \notag\\
\;& - \,   8 \left(k^2 \ell ^2+1\right) \sin k \ell_1 +\left(k^2 \ell ^2+3\right)^2 \sin
   k(\ell_1 +2 \pi ).
\end{align}
We see that the function $c$ is the same as the one in the non-magnetic case \cite[Sec.~2.1]{BET21}, and accordingly, for $k=n\in\mathbb{N}$ where both functions $a$ and $b$ vanish, we arrive at the flat bands of the non-magnetic chain, with the different possibilities described in \cite{BET21}. On the other hand, for the other values of $k\in\mathbb{R}$, one can check directly that the term $a^2+b^2 $ has the same value,
$$128 \sin ^2 k\pi \left(4 k^2 \ell ^2+\left(k^2 \ell ^2+1\right)^2 +\left(k^2 \ell ^2-1\right)^2 \cos 2k(\pi-\ell_3 )\right)>0, $$
for all $A\in\mathbb{Z}$ including, in particular, the non-magnetic case, $A=0$. Thus the band condition \eqref{bandCon,a2b2c2} and the resulting spectral properties of the system are insensitive to integer values of $A$, proving in this way claim \ref{th1}.

\smallskip

Let us consider next the case $ A-\frac 12 \in\mathbb{Z} $ where the functions $a$ and $b$ in \eqref{SC,gen,abc} simplify to
\small
\begin{align}\label{SC-halfintegA,ab}
 a=& \;    8 \cos k\pi \left(  (k \ell -1)^2 \;\sin \left( k(\pi -\ell_3 )-A\ell_3\right)   +(k\ell +1)^2 \;\sin \left(k(\pi -\ell_3 )+A\ell_3   \right)  \right)           ,\notag \\
 b=& \;    8 \cos k\pi \left(   (k \ell -1)^2 \;\cos \left( k(\pi -\ell_3 )-A\ell_3 \right) -(k\ell +1)^2 \;\cos \left(k(\pi -\ell_3 )+A\ell_3  \right)   \right)          .
\end{align}
\normalsize
Consequently, the spectral condition \eqref{formul,abc} has solutions independent of $\,\theta \,$ for $k=n-\frac12$, $n\in\mathbb{N}$, where $a^2+b^2=0$. Indeed, inspecting the function $c$ at these values, the spectral condition \eqref{formul,abc} reduces to
\begin{equation}\label{SC,cForHalf}
\big(  \left(2n-1\right)^{2}\,\ell^2 -4 \big)^2 \, \sin \frac{(2n-1)\ell_1 }{2} \;\sin^2 \frac{(2n-1)\ell_3 }{2} =0,
\end{equation}
from which we infer that
 \begin{itemize}
 \item if $\ell=\left(n-\frac12\right)^{-1}$ and $k=n-\frac12$ with $n\in\mathbb{N}$, the number $k^2$ belongs to the spectrum independently of the other parameters, being always embedded in the continuous spectrum;
   \item  if $\ell_1=2p \pi$,  the number $k^2=\left(n-\frac 12\right)^2$ belongs to the spectrum for any $n,p \in\mathbb{N}$. More generally, assuming that at least one of the edge lengths $\ell_i$, $i=1,3,$ is a rational multiple of $\pi$, namely $\ell_i=2 \frac pq \pi$ with coprime $p,q \in\mathbb{N}$ and odd $q$, we infer that the number $k^2=q^2 \left(n-\frac 12\right)^2$ belongs to the spectrum for all $n,p\in\mathbb{N}$ (in the case of $\ell_3$ being always embedded in the continuous spectrum). Indeed, $a$ and $b$ vanish again and $c$ is given by \eqref{SC,cForHalf} with $2n-1$ replaced by $q(2n-1)$.
   \end{itemize}
Away from the flat bands mentioned, the spectrum is absolutely continuous; in view of the analyticity mentioned above the solution cannot be constant in an open set unless is independent of $\theta$. It has a band-and-gap structure and is the same for all values of $A-\frac12\in\mathbb{Z}$. The band condition \eqref{bandCon,a2b2c2} in this case reads explicitly
\begin{equation}\label{bandCon,Gen,half-int}
128 \cos ^2 k\pi \left(k^4 \ell ^4-(k^2 \ell ^2-1)^2 \cos 2k(\pi-\ell_3 )+6 k^2 \ell ^2+1\right)-(\tau+\rho)^2 \geq 0,
\end{equation}
where
\begin{align*}
& \tau:=2 \left((k^2 \ell ^2-1)^2 \cos 2k(\pi -\ell_3 )-4 (k^2 \ell ^2+1)\right)\sin k\ell_1 ,\\
& \rho:=(k^2 \ell ^2-1)^2 \sin k(2 \pi-\ell_1)-(k^2 \ell ^2+3)^2 \sin k(\ell_1 +2 \pi )  .
\end{align*}
The band-gap pattern for specific examples is illustrated in Figs.~\ref{embedding-ell1} and \ref{embedding-ell3}. \mbox{}\\
 \begin{figure}[!htb]
\centering
\includegraphics[scale=.8]{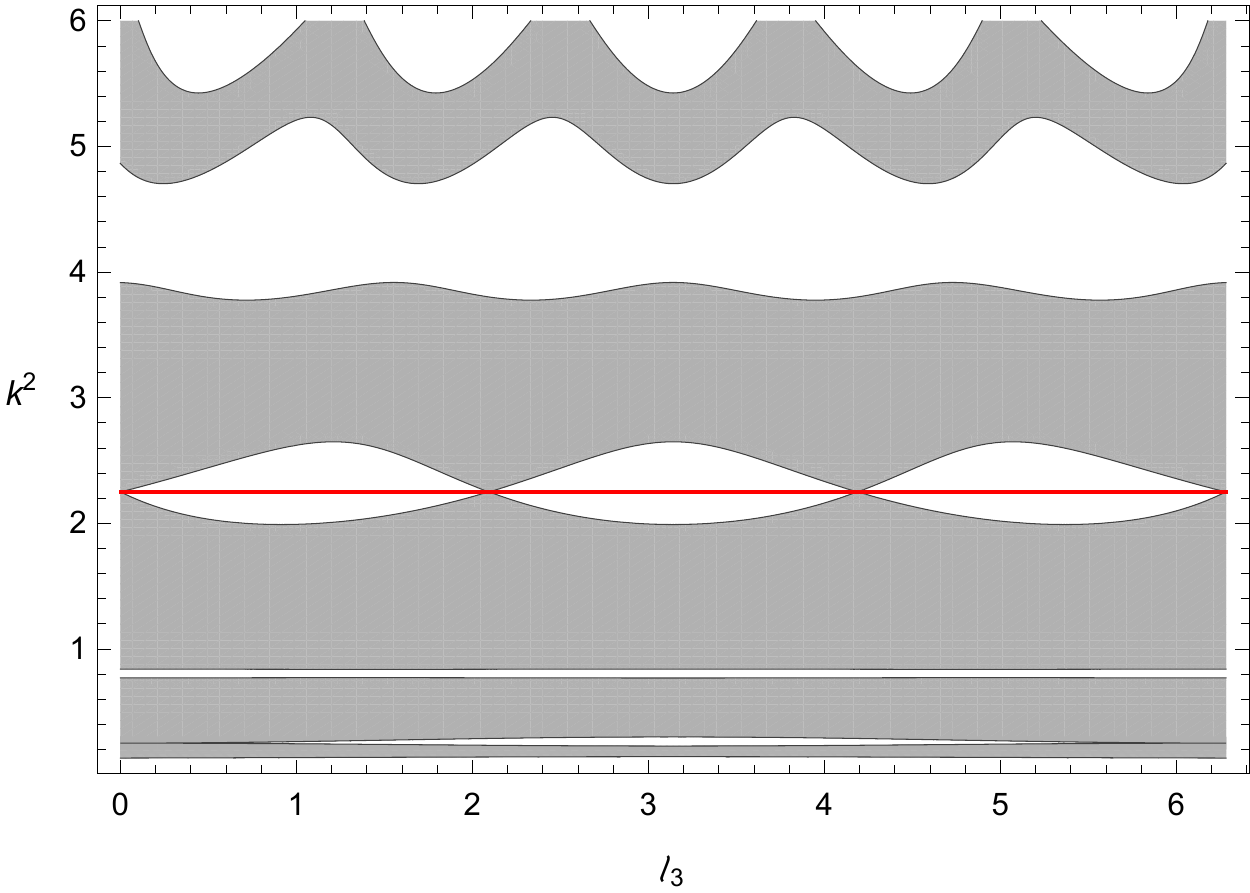}
\caption{Positive spectrum of the  ring array give  by the band condition \eqref{bandCon,Gen,half-int} in dependence on $\ell_3$ for $A-\frac12\in\mathbb{Z}$, $\ell_1=\frac{2}{3}\pi$ and $\ell=1$. The red line marks the flat band $q^2 \left(n-\frac 12\right)^2$ with $q=3$ and $n=1$; the other flat bands referring to higher values of $n$ lay outside the picture area. In this and all subsequent figures, the gray parts represent the spectral bands.   }
\label{embedding-ell1}
\end{figure}
\begin{figure}[!htb]
\centering
\includegraphics[scale=.8]{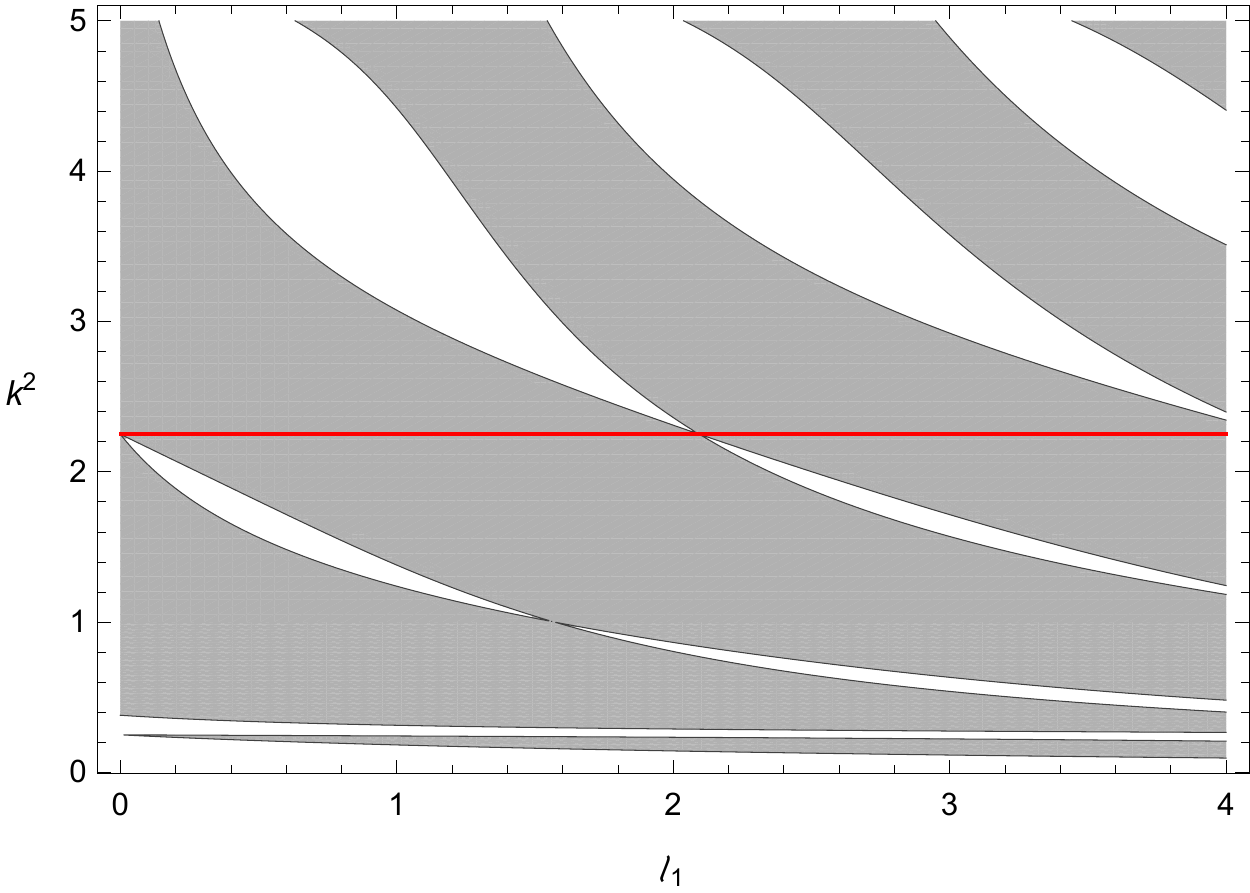}
\caption{The spectrum for the same parameter choice but with $\ell_3=\frac{2}{3}\pi$, now in dependence on $\ell_1$. The red line again corresponds to the flat band $q^2 \left(n-\frac 12\right)^2$ with $q=3$ and $n=1$.}
\label{embedding-ell3}
\end{figure}
\newpage
Now, let us pass to the general case in which the magnetic flux is neither integer nor half-integer, i.e. $2A\notin\mathbb{Z}$. To prove the claim \ref{th3a}, we have to investigate different situations in which both the functions $a$ and $b$ in \eqref{SC,gen,abc} can vanish simultaneously for a general $A\in\mathbb{R}$, specifically: (\emph{a}) the two functions $\sin(A\pm k)\pi$ vanish simultaneously, which can happen either for $A\in\mathbb{Z}$ together with $k=n\in\mathbb{N}$, or for $A-\frac{1}{2}\in\mathbb{Z}$ together with $k=n-\frac{1}{2}\in\mathbb{N}$; (\emph{b}) the four functions $\cos (A\pm k)(\pi -\ell _3)$ and $\sin (A\pm k)(\pi -\ell _3)$ vanishing simultaneously, which is clearly not possible; and furthermore, (\emph{c}) $\sin(A+k)\pi=0$ when both the functions $a$ and $b$ vanish simultaneously at $k=\ell^{-1}$; inspecting then $c$ at these values, one finds that for all $A+\ell^{-1}\in\mathbb{Z}$ and $k=\ell^{-1}$ the number $k^2$ belongs to the spectrum independently of the other parameters.

\smallskip

That, however, does not exhaust all situations in which $a$ and $b$ vanish simultaneously. Consider first the asymmetric situation, $\ell_3\ne\pi$. By simple manipulations, the equations $a=0$ and $b=0$ can be, respectively, rewritten as
$$
\frac{(k \ell -1)^2 \sin (A-k)\pi}{(k \ell +1)^2 \sin (A+k)\pi}=\frac{\cos (A-k)(\pi -\ell_3)}{\cos (A+k)(\pi -\ell_3)}, $$
and
$$ \frac{(k \ell -1)^2 \sin (A-k)\pi}{(k \ell +1)^2 \sin (A+k)\pi}=\frac{\sin (A-k)(\pi -\ell_3)}{\sin (A+k)(\pi -\ell_3)}; $$
consequently, a necessary condition for $a$ and $b$ to vanish simultaneously is
\begin{equation}\label{ab=0;con}
 \frac{\cos (A-k)(\pi -\ell_3)}{\cos (A+k)(\pi -\ell_3)}=\frac{\sin (A-k)(\pi -\ell_3)}{\sin (A+k)(\pi -\ell_3)}.
\end{equation}
Note that those values of $A$ which have been excluded from the consideration here, i.e. the ones in which the denominators of the fractions in \eqref{ab=0;con} vanish, give the same result as in the case (\emph{c}) mentioned above. On the other hand, one can easily check that the necessary condition \eqref{ab=0;con} is fulfilled for $k=\frac{m\pi}{2(\pi -\ell_3)}$ with $m\in\mathbb{Z}\,$; substituting this into $a$ and $b$ in \eqref{SC,gen,abc}, we arrive at
$$ \quad a=  -2 \, i^{3 m+1} \, \Lambda^+\,\sin A (\ell_3-\pi )  , \qquad b=  2 \, i^{3 m+1} \, \Lambda^+ \, \cos A (\ell_3-\pi )  , $$
for odd values of $m$, and
$$ a=  2 \, i^{3 m}\, \Lambda^-\,\cos A (\ell_3-\pi )  , \;\;\quad\qquad b=  2 \, i^{3 m} \, \Lambda^- \, \sin A (\ell_3-\pi ) , $$
for even values of $m$, where
\small
\begin{equation}\label{Lambda}
\Lambda^\pm:=  \left(   \frac{(m \ell +2)\pi-2 \ell_3}{\ell_3-\pi}   \right)^2 \sin \left(A+\frac{m\pi}{2 (\pi-\ell_3) }\right)\pi\pm\left(  \frac{2 \ell_3+  (m \ell -2)\pi}{\ell_3-\pi}  \right)^2 \sin \left(A-\frac{m\pi}{2 \pi-\ell_3) }\right)\pi.
\end{equation}
\normalsize
As a result, the functions $a$ and $b$ can vanish simultaneously if and only if the function $\Lambda^+$, or $\Lambda^-$ vanishes (depending on the parity of $m$). For convenience, let us rewrite \eqref{Lambda}, again, in the compact form
\begin{equation}\label{Lambda,compact}
\Lambda^\pm=  4\left(   k\ell+1\right)^2 \sin \left(A+k\right)\pi\pm  4\left(   k\ell-1\right)^2 \sin \left(A-k\right)\pi  ,
\end{equation}
where $k=\frac{m\pi}{2(\pi -\ell_3)}$. As already mentioned, the functions $k\mapsto\sin (A\pm k)\pi$ cannot vanish simultaneously unless $2A\in\mathbb{Z}$. Manipulating \eqref{Lambda,compact}, by expanding the sine functions of the composed argument, the equations $\Lambda^{+}=0$ and $\Lambda^{-}=0$ can be, respectively, rewritten as
\begin{equation}\label{Lambda,oddm}
\tan A\pi= -\frac{2k\ell}{k^2 \ell ^2+1}\,\tan k\pi,
\end{equation}
and
\begin{equation}\label{Lambda,evenm}
\tan A\pi= -\frac{k^2 \ell ^2+1}{2k\ell}\,\tan k\pi,
\end{equation}
being, respectively again, satisfied for
\begin{equation}\label{A,oddm}
A=-\frac{1}{\pi}\,\arctan\left(\frac{2k\ell}{k^2 \ell ^2+1}\,\tan k\pi\right)+m',
\end{equation}
and
\begin{equation}\label{A,evenm}
 A=-\frac{1}{\pi}\,\arctan\left(\frac{k^2 \ell ^2+1}{2k\ell}\,\tan k\pi\right)+m',
\end{equation}
with $m'\in\mathbb{Z}$. Note that $\tan A\pi$ is periodic with the period $T=1$, it is not defined at $A=\frac12$, and it is increasing for both $A\in(0,\frac12)$ and $A\in(\frac12,1)$ where its range lies in the interval $(0,\infty)$, and $(-\infty,0)$ in the former and latter domains, respectively. On the other hand, for any given values of parameters $\ell>0$ and $k>0$ (or, equivalently, those $\ell_3$ for which $2k=\frac{m\pi}{\pi -\ell_3}\notin\mathbb{N}$), the functions on the right-hand side of \eqref{Lambda,oddm} and \eqref{Lambda,evenm} are always positive or negative constants; as a result, over the domain $A\in(0,1)$, they may cross the function $\tan A\pi$ only once.

In the symmetric case, $\ell_2=\ell_3=\pi$, the function $b$ in \eqref{SC,gen,abc} vanishes identically and the functions $a$ and $c$ simplify to the form
\begin{align}\label{gem,sym,ac}
a=& \; 16 \left( \left(k^2 \ell ^2+1\right) \sin k\pi \, \cos A\pi+2 k \ell \,  \sin A\pi \, \cos k\pi\right), \\[5pt] \nonumber
c=& \;  \left(k^2 \ell ^2+3\right)^2 \sin k(\ell_1 +2 \pi )-2  \left(\left(k^2 \ell ^2-1\right)^2+4  \left(k^2 \ell ^2+1\right)\cos 2 A\pi\right) \sin k\ell_1    \\ \nonumber
& +16 \,k \ell  \,\sin 2A\pi \; \cos k\ell_1-\left(k^2 \ell ^2-1\right)^2 \sin k(2 \pi -\ell_1 )\,,  \nonumber
\end{align}
from which claims \ref{th1}--\ref{th3a} of \autoref{thmGen} follow easily. As for the claim \ref{th3b}, one obtains exactly the same condition as \eqref{Lambda,evenm} by manipulating the equation $a=0$, in this case for the general $k\in\mathbb{R}$.

\smallskip

Away from those exceptional cases mentioned, the rest of the spectrum is again absolutely continuous having a band-and-gap structure; the general band condition is given by \eqref{bandCon,a2b2c2} together with \eqref{SC,gen,abc}. To provide an illustration, we show two parameter dependencies in Figs.~\ref{figGen1,ell3,pi3} and \ref{figGen2,ell1,pi3}; we include here also negative part of the spectrum given by the condition \eqref{gen-Neg,abc} which will be discussed below.
\begin{figure}[!htb]
\centering
\includegraphics[scale=.8]{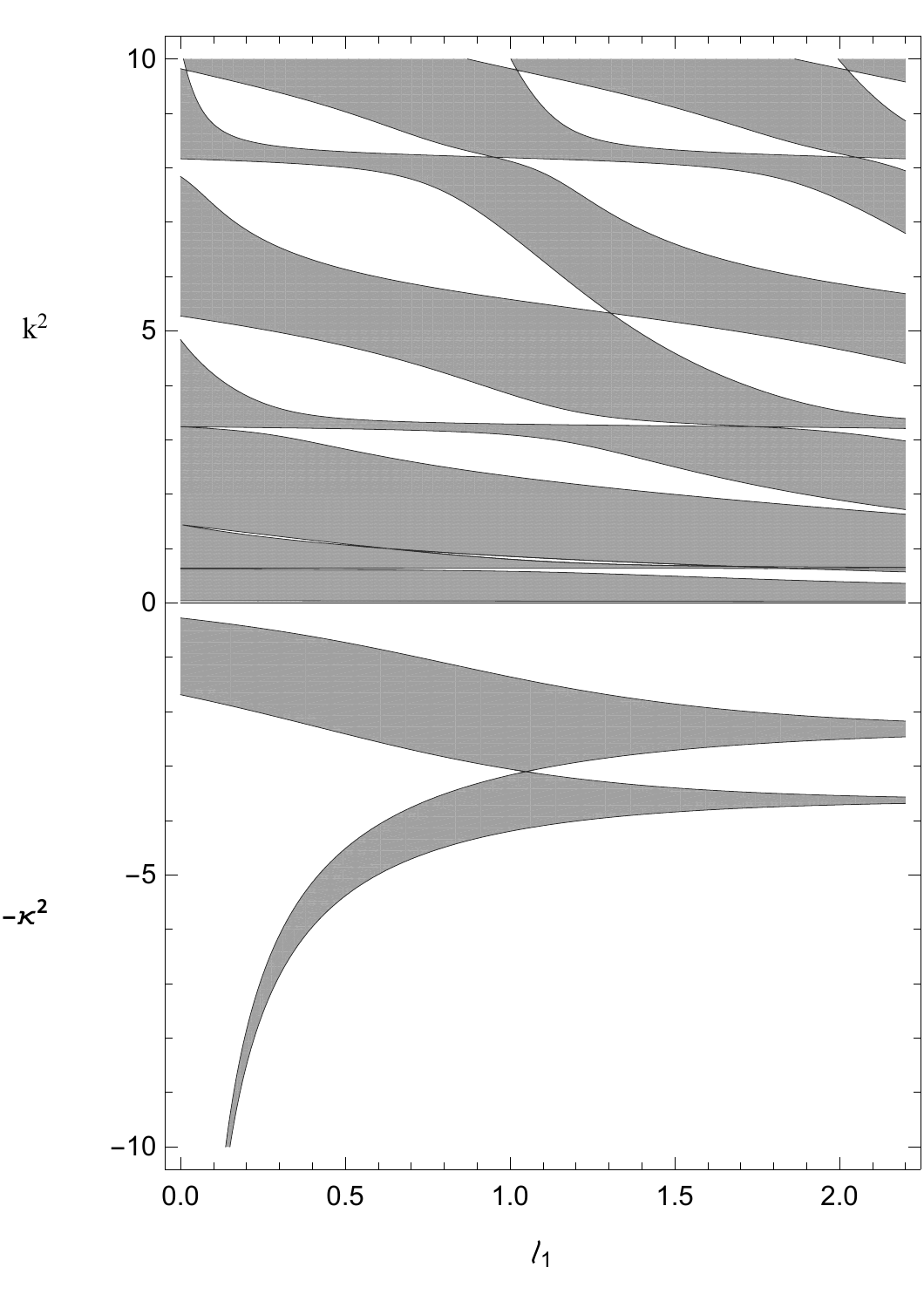}
\caption{Spectrum in dependence on $\ell_1$ for $\ell_3=\frac{\pi}{3}$, $\ell=1$, and $A=\frac 15$.}
\label{figGen1,ell3,pi3}
\end{figure}
\begin{figure}[!htb]
\centering
\includegraphics[scale=.8]{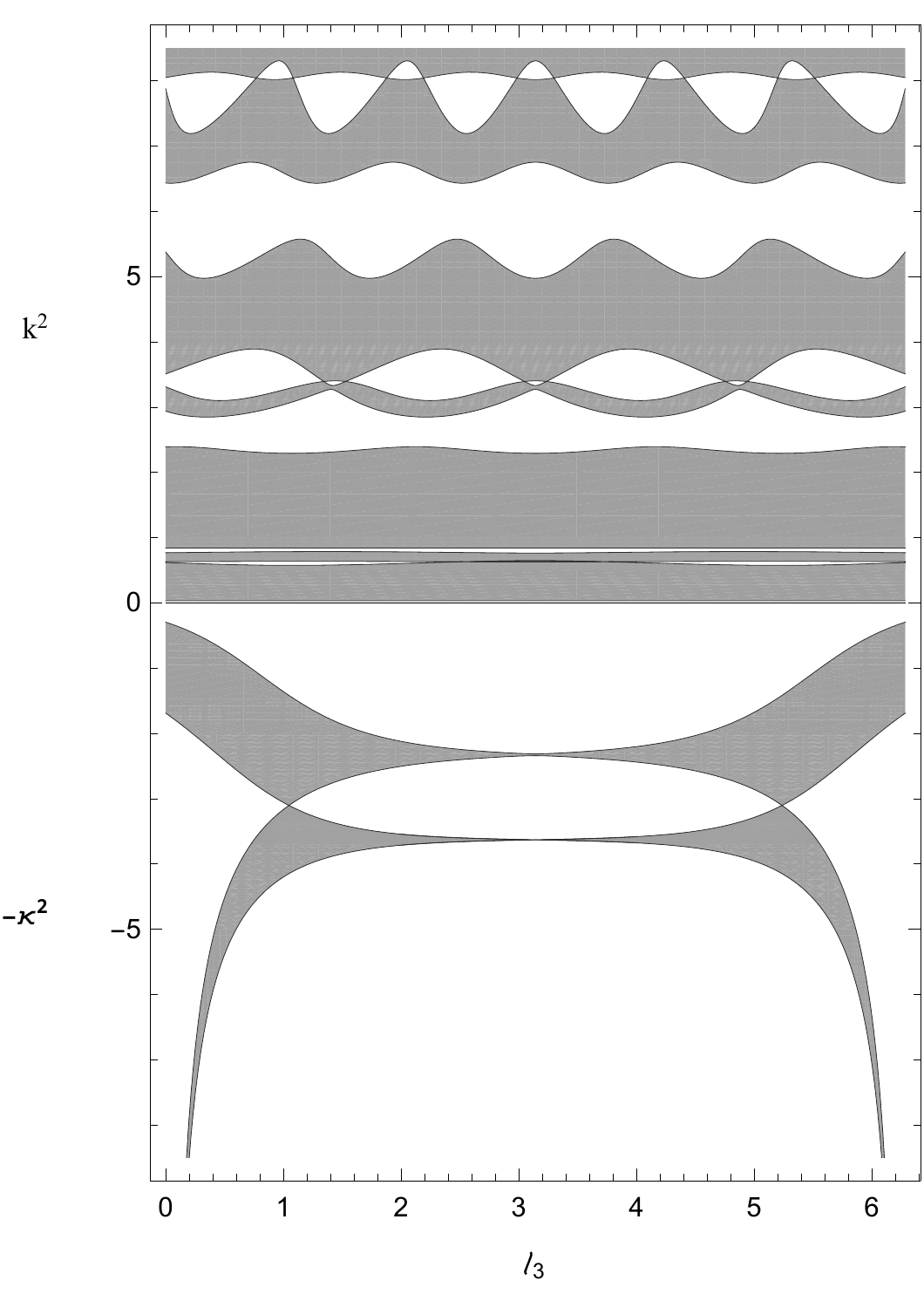}
\caption{Spectrum in dependence on $\ell_3$ for the same parameter choice but with $\ell_1=\frac{\pi}{3}$; the symmetry with respect to interchange of $\ell_3$ and $2\pi-\ell_3$ is obvious.}
\label{figGen2,ell1,pi3}
\end{figure}
As a pendant to Fig.~\ref{figGen1,ell3,pi3} we show in Fig.~\ref{fig-symm-ell=2} the spectral pattern in dependence on $\ell_1$ in the symmetric case $\ell_3=\pi$.
\begin{figure}[!htb]
\centering
\includegraphics[scale=.8]{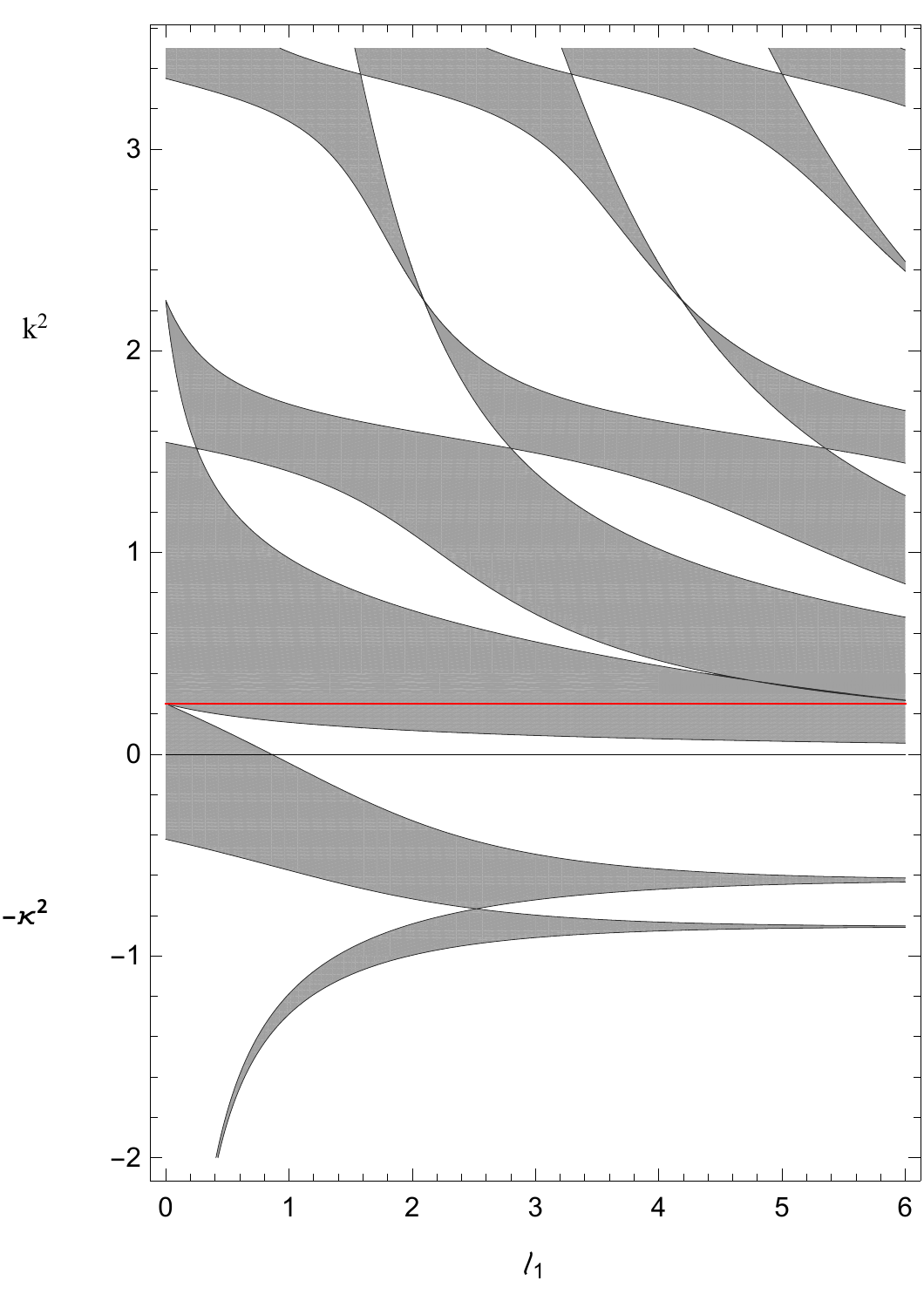}
\caption{Spectrum of the symmetric array model for $\ell=2$ and $A-\frac 12\in\mathbb{Z}$ in dependence on $\ell_1$. The red line at energy $\frac14$ corresponds to the flat band of the first bullet point following eq.~\eqref{SC,cForHalf} with $n=1$; the values $\ell_1=\left\{\frac{2\pi}{3},\frac{4\pi}{3}\right\}$ with energy $k^2 = \frac94$, for which the bands shrink to points, correspond to the flat bands of the second bullet point there.}
\label{fig-symm-ell=2}
\end{figure}

Before proceeding to claim~\ref{th5} of the theorem, let us take a closer look at the properties of the positive part of the spectrum; inspecting the spectral condition we see that
\begin{itemize}
\item The spectrum is symmetric with respect to the swap of $\ell_3$ and $\ell_2=2\pi-\ell_3$; note that although $\ell_3$ appears as the argument of both sine and cosine functions in $a$, $b$, and $c$ in \eqref{SC,gen,abc}, the symmetry arises from the band condition \eqref{bandCon,a2b2c2} which contains them in the squared form.
\item As in the non-magnetic case, it may happen that some of the gaps close when the boundaries of neighboring bands touch as can be seen in Figs.~\ref{embedding-ell1}--\ref{figell10} (as well as Fig.~\ref{figell30} in the degenerate situation); such crossings appear in sequences with the same energy. Denoting the left-hand side of \eqref{SC,gen,ell-1,3} as $\Delta$, sufficient conditions to have such crossings are given by the relations
\begin{equation}\label{cross,con}
\frac{\partial \Delta}{\partial \ell_i}=\frac{\partial \Delta}{\partial \theta}=\frac{\partial \Delta}{\partial k}=0,\qquad i=1,3,
\end{equation}
provided that $\ell$ and one of the lengths $\ell_i$ are fixed, in particular, one can identify their closed-form coordinates under commensurability conditions, $\ell_i=\frac pq \pi$, similarly as in \cite{BET21}.
\item The dominating property of the considered class of vertex couplings, namely that the transport vanishes at the high-energy limit if the vertex parity is odd, is not affected by the presence of the magnetic field. As the vertex degree is three in our case, one expects that the spectrum will be dominated by gaps at high energies. This is indeed the case, as we can check by keeping the highest power of $k$ in the band condition \eqref{bandCon,a2b2c2} --~together with \eqref{SC,gen,abc}~-- which is the same as in the non-magnetic case, that is,
\begin{equation}
-16\;k^8 \ell ^8\;\sin ^2 k \ell _1 \;\sin ^2 k \left(2 \pi -\ell _3\right) \;\sin ^2 k \ell _3 +\mathcal{O}(k^6) \geq 0 ,
\end{equation}
indicating that the spectral bands may exist only in the vicinity of the points $k_j=\frac{n\pi}{\ell_{j}}$ with $j=1,2,3$ and $n\in\mathbb{N}$, recall that $\ell _2+\ell _3=2\pi$. In other words, \emph{the probability of belonging to the spectrum} for a randomly chosen value of $k$, introduced by Band and Berkolaiko \cite{BB13} as
\begin{equation}\label{probsigma}
P_{\sigma}(H):=\lim_{K\to\infty} \frac{1}{K}\left|\sigma(H)\cap[0,K]\right|,
\end{equation}
equals zero, and not only for incommensurate edges.
\item The effect of the magnetic field is more subtle. In contrast to the case $A=0$, for instance, it may happen that the first positive band remains separated from zero. To see that, we use Taylor expansion around $k=0$ in the band condition \eqref{bandCon,a2b2c2} and \eqref{SC,gen,abc},
$$ -8 \sin ^2 A\pi  \left(   \ell_1 ^2+4 \pi  (\ell_1 +\ell_3 )+\cos 2A\pi \left(4\ell ^2-\ell_1 ^2\right)+4 \ell_1 \ell  \sin 2A\pi-2 \ell_3 ^2-4 \ell^2   \right)k^2+ \mathcal{O}(k^4) , $$
where the leading term may be for a non-integer $A$ negative for some parameters values. In particular, one can check that for small values of $\ell$, the expression in the large brackets reads
$$ \left(  \ell_1^2\left(1-\cos2A\pi\right) +2\ell_3\left(2\pi-\ell_3\right)+4\pi\ell_1 \right) + \mathcal{O}(\ell) $$
with the leading term being always positive, showing that small values of $k$ may not belong to the spectrum.
\end{itemize}

\smallskip

After this interlude, let us turn now to the \emph{negative part} and conclude thus the proof of Theorem~\ref{thmGen}. It can be obtained from \eqref{SC,gen,ell-1,3}, or equivalently from \eqref{SC,gen,abc} in which we replace $k$ by $i\kappa$ with $\kappa>0$. This yields the spectral condition \eqref{formul,abc} where this time we have
\begin{align}\label{gen-Neg,abc}
a=\;&     -4 (\kappa ^2 \ell ^2-1 ) \left(\cos A \ell_3 \, \sinh \kappa(2 \pi -\ell_3 ) +\cos A (2 \pi -\ell_3 ) \, \sinh \kappa\ell_3  \right)      \notag\\ \nonumber
\;& + \, 8 \kappa  \ell \left( \sin A (2 \pi -\ell_3 )\, \cosh  \kappa\ell_3  +  \sin A\ell_3 \,\cosh \kappa(2 \pi -\ell_3 )    \right) , \\[10pt] \nonumber
b=\;&  4 (\kappa ^2 \ell ^2-1) \left(\sin A (2 \pi -\ell_3 )\, \sinh \kappa\ell_3 -\sin A \ell_3 \, \sinh \kappa(2 \pi -\ell_3 ) \right)         \notag \\
\;& + \, 8 \kappa\ell \left( \cos A (2 \pi-\ell_3 )\, \cosh \kappa\ell_3 -  \cos A\ell_3 \, \cosh \kappa(2 \pi -\ell_3 ) \right),  \\[10pt]
c=\;&   \left(     4 (\kappa ^2 \ell^2-1)  \cos 2 A\pi  +  (\kappa ^4 \ell ^4+3 ) \left(\cosh 2 \kappa\pi-\cosh 2 \kappa(\pi -\ell_3 )  \right)     \right)   \sinh \kappa\ell_1         \notag\\
\;& + \, 2   \left(4 \kappa  \ell  \sin 2A\pi -(\kappa ^2 \ell ^2-1) \left( \sinh 2\kappa (\pi-\ell_3 ) +\sinh 2 \kappa\pi    \right)\right) \cosh \kappa\ell_1       \notag\\
\;& - \, 4\left(\kappa ^2 \ell ^2-1\right) \cosh \kappa(2 \pi -\ell_3 ) \, \sinh \kappa(\ell_1 +\ell_3 )  .\notag
\end{align}
As in the case of positive energies, a number $-\kappa^2$ belongs to a spectral band if it satisfies the band condition \eqref{bandCon,a2b2c2}, now with the input from \eqref{gen-Neg,abc}. The negative spectrum has the following properties:
\begin{itemize}
  \item There is no flat band in the negative part of the spectrum; for $A-\frac 12 \in\mathbb{Z}$, the functions $a$ and $b$ in \eqref{gen-Neg,abc} contain the multiplicative factor $\cosh \kappa\pi$ which is non-zero for all $\kappa>0$.
  \item As discussed and proved in \cite{BET21}, there are at most two negative bands which may merge at one point; in the general case, we are unable to find the crossing coordinate in a closed form, however, one can check this crossing employing a relation analogous to \eqref{cross,con}. In particular, in the case of non-magnetic \emph{symmetric} chain, as illustrated in \cite{BET20}, this crossing happens at $\ell_1=\pi$; however, this is no longer true in the magnetic case, cf. Fig.~\ref{fig-symm-ell=2}.
\item As in the non-magnetic case, in the limit $\ell_1\rightarrow \infty$, the bands shrink to points. Note that the only $\ell_1$-dependent function in \eqref{gen-Neg,abc} is $c$, hence, considering $\sinh \kappa\ell_1 \approx \cosh \kappa\ell_1 \approx \frac {\e^{\kappa\ell_1} }{2}$, the spectral condition \eqref{formul,abc} for a fixed $\kappa>0$ can be expressed in the form
\begin{equation}\label{large,ell-1,gen}
-f(\ell,\ell_3,A;\kappa)\, \e^{\kappa\ell_1}+\left( a \cos\theta+b\sin\theta \right)+\mathcal{O}(\e^{-\kappa\ell_1}) =0,
\end{equation}
where
\begin{align}\label{large,f,gen}
f(\ell,\ell_3,A;\kappa):&=4  (\kappa ^2 \ell ^2-1 ) \left(\cos 2A\pi-\sinh 2\kappa\pi \right)+ (\kappa ^4 \ell ^4-2 \kappa ^2 \ell^2+5)\cosh 2\kappa\pi  \\ \notag
& \; +8 \kappa\ell \; \sin 2A\pi-(\kappa ^2 \ell^2+1)^2 \cosh 2\kappa(\pi -\ell_3 )  , \notag
\end{align}
implying that in the indicated limit the bands shrink to the points determined by the equation $f(\ell,\ell_3,A;\kappa)=0$.
  \item In contrast to the non-magnetic case \cite{BET21}, for $A\notin\mathbb{Z}$ and particular values of parameters, it may happen that the negative band starts at zero. One can check this by taking the Taylor expansion of the band condition \eqref{bandCon,a2b2c2} --~together with \eqref{gen-Neg,abc}~-- around $\kappa\rightarrow 0+$; for the sake of simplicity, let us evaluate this for the particular case of $A-\frac 12\in\mathbb{Z}$ for which, we have
      $$-64 \left(  -4 \ell ^2+ \left(\ell_1^2 -\ell_3^2\right)+2 \pi  \left(\ell_1 +\ell_3 \right) \right)\kappa^2+\mathcal{O}(\kappa^4), $$
      where the leading term may be positive for particular values of parameters, implying that small values of $\kappa$ may belong to the spectral bands.
  \item As a particular case of the above point, in contrast to the non-magnetic symmetric case, it may happen that the interval $(-\ell^{-2},0)$ belongs to the spectrum for particular values of parameters. Again, for the sake of simplicity, we show this for the special case $A=m-\frac 12$, $m\in\mathbb{Z}$, for which we can rewrite the negative spectral condition in the form
\small
\begin{align}\label{SymInt-1ell-Neg}
&  \cos\theta=g(\kappa) , \\[7pt]
& g(\kappa):=\frac{\left(\kappa ^2 \ell ^2-3\right)^2 \sinh  \kappa( 2 \pi+\ell_1 ) -\left(\kappa ^2 \ell^2+1\right)^2 \sinh \kappa(2 \pi -\ell_1 )  -2 \left(\kappa ^4 \ell ^4+6\kappa ^2 \ell ^2-3\right) \sinh\kappa\ell_1  }{32 \,  (-1)^m \, \kappa\,\ell  \cosh \kappa\pi  }. \notag
\end{align}
\normalsize
Now, by an example, we show that this equation may have solutions for $\kappa\in(0,\ell^{-1})$; inspecting $g(\kappa)$ at the specific value $\kappa=\frac {1}{2} \ell^{-1}\in(0,\ell^{-1})$, we arrive at
$$g\left(\frac {1}{2\ell}\right)=\frac{(-1)^m}{8}\;\text{sech} \frac{\pi }{2 \ell }  \left(3 \sinh \frac{\pi }{\ell}  \cosh  \frac{\ell_1 }{2 \ell } +\left(\frac{73}{16} \cosh \frac{\pi }{\ell } +\frac{23}{16}\right) \sinh \frac{\ell_1 }{2\ell } \right) , $$
which is monotonous with respect to increasing $\ell_1$ for a fixed value of $\ell$ in view of the fact that both functions $\sinh x$ and $\cosh x$ are positive and monotonously increasing with respect to $x>0$. Consequently, the function reaches its minimum or maximum --~depending on the parity of $m$~-- at $\ell_1=0$ with the value $\pm \frac{3}{4} \sinh \frac{\pi }{2 \ell } $, the upper and lower signs correspond to even and odd $m$, respectively. As a result, by increasing $\ell$, the extremum tends to zero, which means that $g(\kappa)$ may lay in the interval $[-1,1]$.
\end{itemize}
\end{proof}

\section{The case $\ell_1=0$}
\label{sect:Mod-ell1=0}
Let us now pass to degenerate versions of the model when one of the edge lengths shrinks to zero, starting with the case of a tightly coupled array in which the rings touch directly, $\ell_1=0$. The vertices are then of degree four, and following the observation made in \cite{ET18}, we expect that this will have a profound influence on the behavior of the spectrum at high energies.

\subsection{Positive spectrum}
\label{sect:Mod-ell1=0,pos}
In this case, $a$, $b$, and $c$ in the spectral condition \eqref{formul,abc} are as follows,
\small
\begin{align}\label{SC,ell-1=0,abc}
a=& \;(k \ell +1)^2 \sin(A+k)\pi \;\cos (A-k)(\pi -\ell _3)-(k \ell -1)^2 \sin (A-k)\pi \;\cos (A+k)(\pi -\ell _3),\notag\\[5pt]
b=& \;(k \ell -1)^2 \sin (A-k)\pi \;\sin (A+k)(\pi -\ell _3)-(k \ell +1)^2 \sin(A+k)\pi \;\sin (A-k)(\pi -\ell _3) ,\notag\\[5pt]
c=&  \; 2 k \ell  \sin 2 A\pi + (k^2 \ell^2+1) \sin 2 k\pi .
\end{align}
\normalsize
In the particular case of $A-\frac 12\in\mathbb{Z}$, these functions simplify to
\begin{align}\label{SC,ell-1=0,abc,half-A}
a=& \;  \cos k\pi \left(2 \left(k^2 \ell ^2+1\right) \cos A \ell_3 \, \sin k(\pi -\ell_3 ) +4 k \ell \, \sin A \ell_3 \,  \cos k(\pi -\ell_3 )  \right)      ,\notag\\[5pt]
b=& \;   \cos k\pi \left(  2 \left(k^2 \ell ^2+1\right) \sin A \ell_3 \, \sin k(\pi -\ell_3 )-4 k \ell \, \cos A \ell_3  \, \cos k(\pi -\ell_3 )    \right)      ,\notag\\[5pt]
c=& \;  2(k^2 \ell ^2+1) \sin k\pi\;\cos k\pi,
\end{align}
indicating that, again, the flat bands in the half-integer flux situation occur at $k=n-\frac 12$ with $n\in\mathbb{N}$, in this case, for all $\ell_3>0$. Excluding those flat bands, the spectrum is for  $A-\frac12\in\mathbb{Z}$ continuous having a band-and-gap structure; the band condition \eqref{bandCon,a2b2c2} in this case reads explicitly
\begin{equation}\label{band-gen-halfInt-A}
 4 k^2 \ell ^2 +\left(k^2 \ell^2+1\right)^2 \cos 2k\pi -\left(k^2 \ell ^2-1\right)^2 \cos 2k(\pi -\ell_3 ) \geq 0  .
\end{equation}
In the generic flux case, $2A\notin\mathbb{Z}$, the claim \ref{th3a} of Theorem~\ref{thmGen} is still valid, which can be proven in the same way as in Sec.~\ref{sect:GenMod} since the functions $a$ and $b$ are the same as those in the general model, however, the claim \ref{th3b} does not hold in this case; one can check this directly by substituting \eqref{A,oddm} and \eqref{A,evenm}, respectively, into $c$ given by \eqref{SC,ell-1=0,abc}, which after simplifications yields
$$   c=\frac{\left(k^2 \ell ^2-1\right)^2 \left(k^2 \ell ^2+1\right) \sin 2k\pi}{\left(k^2 \ell ^2+1\right)^2+4 k^2 \ell ^2 \tan^2 k\pi },  $$
and
$$   c= \frac{2 \left(k^2 \ell ^2-1\right)^2 \left(k^2 \ell ^2+1\right) \sin^2 k\pi \; \tan k\pi}{4 k^2 \ell ^2+\left(k^2 \ell ^2+1\right)^2 \tan ^2 k\pi} .   $$
As can be seen, these functions $c$ may vanish only for $k=\ell^{-1}$ and $2k\in\mathbb{N}$ which is the situation already discussed, in other words, there are no bands shrinking to points for particular parameter values in this case.

\smallskip

Away from the flat bands $\ell^{-2}$ which occur for $A+\ell^{-1}\in\mathbb{Z}$, the spectrum for the values $2A\notin\mathbb{Z}$ is again absolutely continuous having a band-gap structure; the spectral bands are given by the condition \eqref{bandCon,a2b2c2} together with \eqref{SC,ell-1=0,abc}. The band-gap pattern for particular values of parameters is illustrated in Fig.~\ref{figell10}.
\begin{figure}[!htb]
\centering
\includegraphics[scale=1.1]{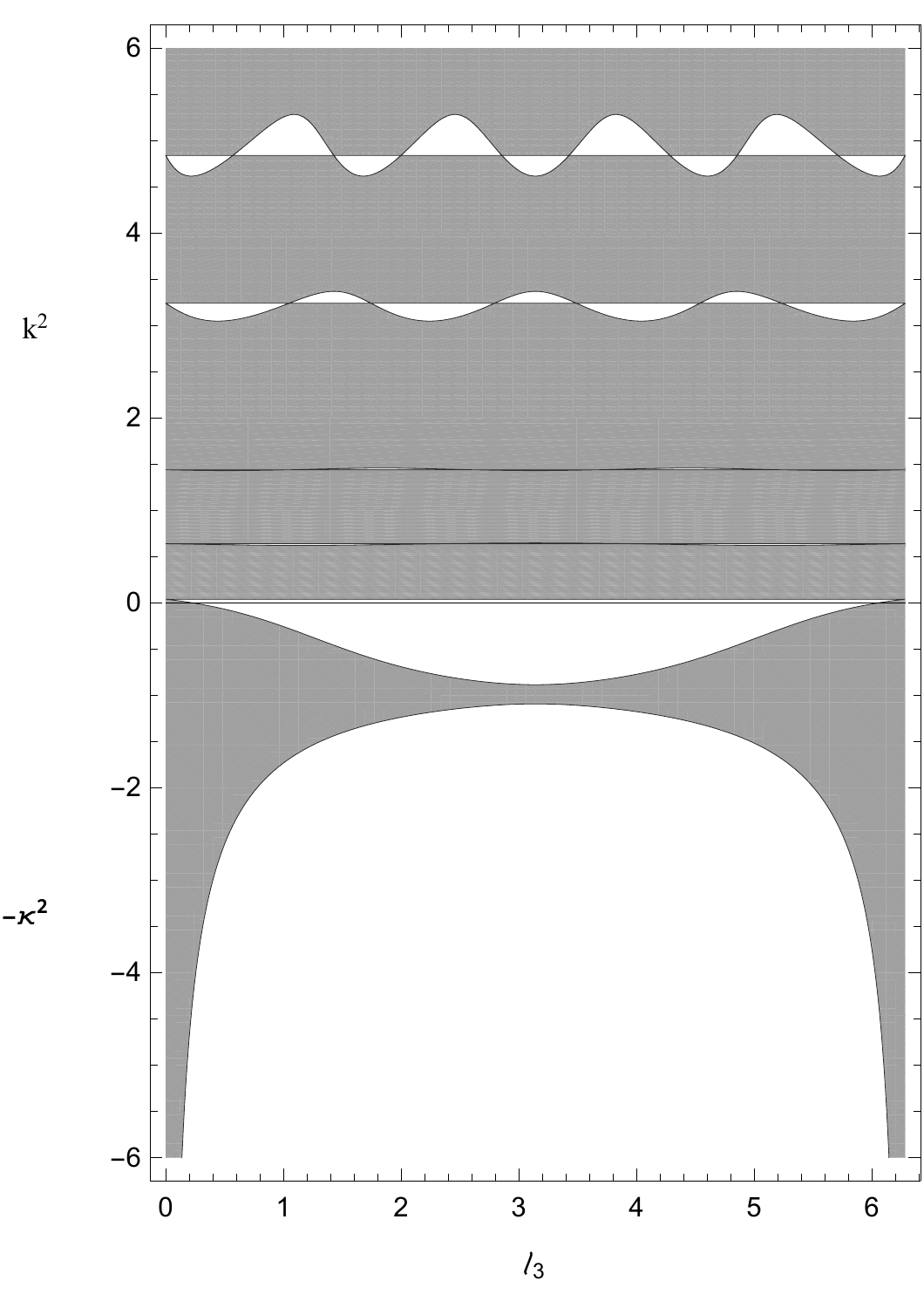}
\caption{ Spectrum of the model with $\ell_1=0$ in dependence on $\ell_3$ for $\ell=1$ and $A=\frac 15$ (the band condition \eqref{bandCon,a2b2c2} together with \eqref{SC,ell-1=0,abc} and \eqref{Neg-ell-1=0,abc}, for the positive and negative spectrum, respectively). Two spectral gaps are present in the interval $\frac12<k^2<\frac32$ but not well seen at the used scale. }
\label{figell10}
\end{figure}

\smallskip

The effect of the magnetic field is especially significant in the symmetric case, $\ell_3=\pi$, where we know from \cite{BET20} that the non-magnetic spectrum fills the whole interval $[0,\infty)$ without any gaps. This is no longer true in the presence of a field with a non-integer flux through the rings, $A\notin\mathbb{Z}$, where the spectral condition reduces to
\begin{equation}\label{ell-1=0,SC-Pos,symm}
\left( 2  \left(k^2 \ell ^2+1\right) \sin k\pi \, \cos A\pi +4 k \ell \, \sin A\pi \, \cos k\pi    \right) \cos\theta = 2 k \ell \, \sin 2A\pi+\left(k^2 \ell ^2+1\right) \sin 2k\pi ,
\end{equation}
which obviously implies that the spectrum has a band-and-gap structure, see Figs.~\ref{fig-symm-ell=2} and \ref{figell10}. It also means that \emph{the probability of belonging to the spectrum} is nontrivial, in contrast to the `loosely connected' chain of the previous section. To investigate the asymptotic behavior of the bands determining when an energy value belongs to the spectrum at high energies, we keep again the leading order term in \eqref{bandCon,a2b2c2} --~this time together with \eqref{SC,ell-1=0,abc}~-- and arrive at the simplified condition
\begin{equation}\label{highCon,ell_1=0}
 \sin(k-A)\pi\;\sin(A+k)\pi\;\sin k\ell_3\;\sin k(2\pi-\ell_3)  \geq 0,
\end{equation}
with the relative error $\mathcal{O}(k^{-2})$. Away of the narrow bands which may appear in pairs in the vicinity of the roots of the sine functions in \eqref{highCon,ell_1=0}, there are wide bands and the gaps between them which grow asymptotically with the band index but not at the same rate in general. We have been unable to find the probability \eqref{probsigma} in a closed form as it takes different values when $\tfrac{\ell_2}{\ell_3}\in\mathbb{Q}$; an example is illustrated in Fig.~\ref{fig,prob,rationalell10} for $A=\frac14$,
\begin{figure}[!htb]
\centering
\includegraphics[scale=1.05]{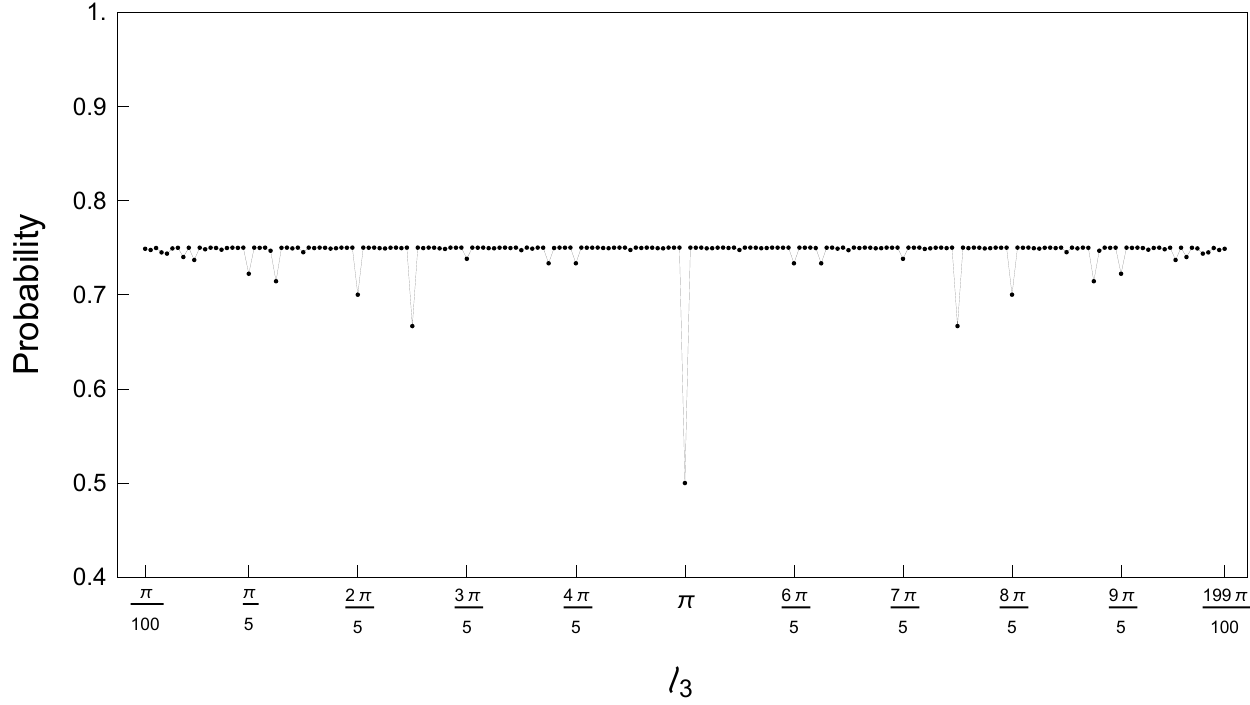}
\caption{The probability \eqref{probsigma} of the model with $\ell_1=0$, determined by \eqref{highCon,ell_1=0}, \emph{vs.} the edge length $\ell_3\in(0,2\pi)$ (or equivalently, $\ell_2$) for $A=\frac 14$. To make the results more visible, we have drawn here, and in the following similar figures, the gray lines joining the black points referring to the adjacent values of the probability.}
\label{fig,prob,rationalell10}
\end{figure}
note that for commensurate $\ell_2$ and $\ell_3$, the function on the left-hand side of \eqref{highCon,ell_1=0} is periodic, and accordingly, one can numerically calculate the probability by finding the roots of the function, and then evaluating the fraction of the period in which the function is non-negative. However, when $\ell_2$ and $\ell_3$ are incommensurate, the probability takes for any fixed $A$ a specific value, varying in the interval $(\frac12,\frac34)$, cf. Fig.~\ref{fig-prob-irrationalell10},
\begin{figure}[!htb]
\centering
\includegraphics[scale=1.3]{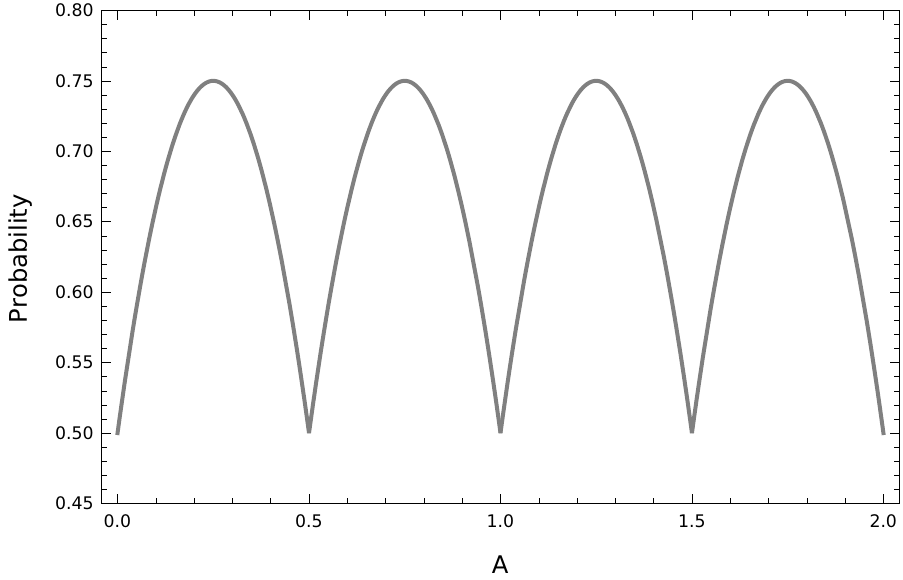}
\caption{The probability \eqref{prob,irrat,ell-1} of the model with $\ell_1=0$ \emph{vs.} the potential value $A$ assuming that $\ell_2$ and $\ell_3$ are incommensurate.}
\label{fig-prob-irrationalell10}
\end{figure}
\emph{independent of specific edge lengths} in accordance with universality property found in \cite{BB13}, as well as of the parameter $\ell$. Note that if $\frac{\ell_2}{\ell_3}\notin\mathbb{Q}$, the terms $k\ell_3$ and $k\pi$ in \eqref{highCon,ell_1=0} behave in the limit $k\rightarrow\infty$ as a pair of independent identically distributed random variables on $[0,2\pi)$, denoted as $(x,y)$. Consequently, one can calculate the probability by finding the fraction of the total area $4\pi^2$ in which
\begin{equation}\label{Prob,Axy,ell_1=0}
 \sin(y-A\pi)\;\sin(y+A\pi)\;\sin x\;\sin (2y-x)\geq 0
\end{equation}
holds. To begin with, we note that the region in which the condition \eqref{Prob,Axy,ell_1=0} is fulfilled is a union of polygons, cf. Fig.~\ref{figProb-ell10-3D};
\begin{figure}[!htb]
\centering
\includegraphics[scale=0.65]{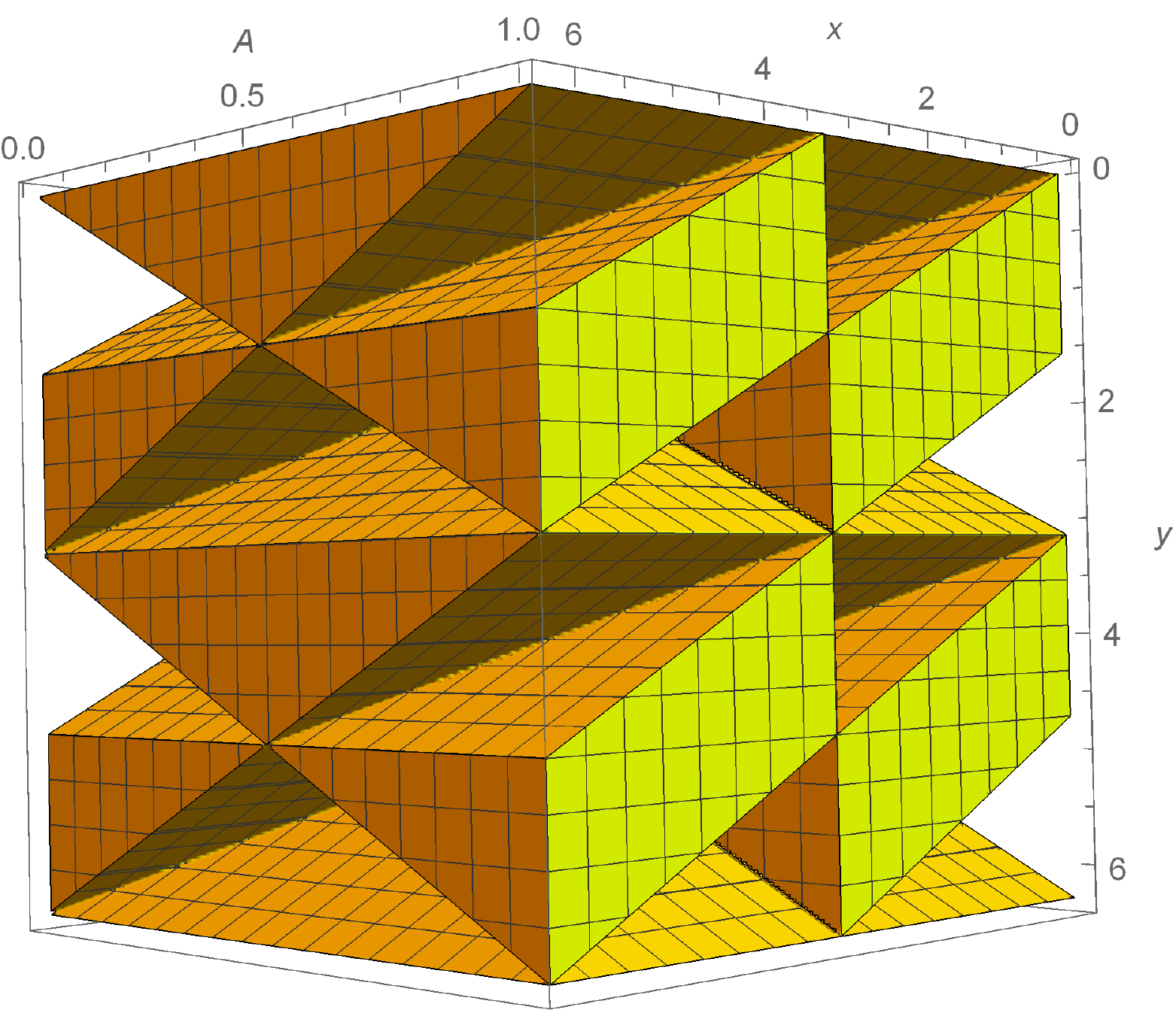}
\caption{The three-dimensional region in which the condition \eqref{Prob,Axy,ell_1=0} holds.}
\label{figProb-ell10-3D}
\end{figure}
on the other hand, one easily checks that the function on the left hand side is symmetric with respect to $x\leftrightarrow x-2\pi$, $y\leftrightarrow y-\pi$, $A\leftrightarrow 2n-A$, and $A\leftrightarrow 2\left(n-\frac12\right)-A$ with $n\in\mathbb{N}$. Accordingly, to find the region in which the function is non-negative, it suffices to find the non-negative part in the first octant only, i.e. in the intervals $x\in[0,\pi)$ and $y\in[0,\frac{\pi}2)$. This can be done analytically by computing the area of the trapezoids enclosed by the lines $y=\frac x2$, $y=A\pi$ with $A\in[0,\frac12]$, and the borders of the intervals; multiplying then the obtained result by eight, and dividing it next by the total area $4\pi^2$, we arrive at the probability
\begin{equation}\label{prob,irrat,ell-1}
P_{\sigma}(H)=\frac12+2A-4A^2 \quad \text{with} \quad (A \bmod\tfrac12)
\end{equation}
for all $A\in\mathbb{R}$, where we have taken into account the invariance with respect to $A$ mentioned above. We see that the second derivative of the probability function with respect to $A$ is $-8<0$, and accordingly, the probability \eqref{prob,irrat,ell-1} as a function of $A$ is strictly concave with its maximum and minima in the interval $A\in[0,\frac12]$ occurring at $A=\frac14$ and $A=0,\frac12$, respectively. Extending the claim to the whole domain $A\in\mathbb{R}$, the maxima and minima of the probability occur, respectively, at $A=\frac{2m-1}{4}$ and $A=0,\frac{2m-1}{2}$ with $m\in\mathbb{Z}$, taking the extremum values $\frac 34$ and $\frac 12$. This is consistent, of course, with the result in the non-magnetic case obtained in \cite{BET21} (for an asymmetric chain which is certainly true if $\ell_2$ and $\ell_3$ are incommensurate) but the same happens if the magnetic flux through the rings takes a half-integer value, cf. the plot of probability \eqref{prob,irrat,ell-1} versus $A$ given in Fig.~\ref{fig-prob-irrationalell10}.

It is also worth noting that for a rational number $\frac{\ell_2}{\ell_3}=\frac pq$ the universality does not hold, however, if $p,q$ are large coprime as a rational approximation of an irrational number, the probability \eqref{prob,irrat,ell-1} is close to the universal value. This is visible in Fig.~\ref{fig,prob,rationalell10} where for $\frac{\ell_2}{\ell_3}\in\mathbb{Q}$ with larger $p$ and $q$ and $A=\frac14$, the value is $\approx 0.749$,  very close to $\frac34$ of the incommensurate edge lengths situation.

\smallskip

It is also interesting to investigate the probability of belonging to the spectrum for the \emph{symmetric} chain. In this case, of course, makes no sense to speak about the universality in the sense of \cite{BB13} but the quantity shows the dependence of the spectral measure on the magnetic field; in addition, it is again independent of the parameter $\ell$ in \eqref{coupB}. The inequality \eqref{highCon,ell_1=0} for $\ell_3=\pi$ reduces to $\sin(k-A)\pi\;\sin(A+k)\pi\;\sin^2 k\pi \geq 0 $, and accordingly, it remains to calculate the probability that the expression
$$   h(k):=\sin(k-A)\pi\;\sin(A+k)\pi   $$
is non-negative for a randomly chosen value of $k$. The function $h(\cdot)$ is periodic with the period $1$; by a straightforward computation, one finds that the roots of the function in the period are
\begin{align*}
& k_1=\dfrac{1}{2\pi}\,\arccos\left(\cos2A\pi\right) ,\\
& k_2=1-\dfrac{1}{2\pi}\,\arccos\left(\cos2A\pi\right).
\end{align*}
On the other hand $h(0)=h(1)=-\sin^2 A\pi<0$, which means that the function is non-negative over the interval $(k_1,k_2)$, and consequently, the probability of belonging to the spectrum for all $A\in\mathbb{R}$ is obtained as
\begin{equation}\label{prob,irrat,ell-1,sym}
P_{\sigma}(H)=1-\dfrac{1}{\pi}\,\arccos\left(\cos2A\pi\right).
\end{equation}
Note that since the range of $\arccos$ function is contained in the interval $[0,\pi]$, the probability of belonging to the spectrum for the symmetric chain varies in the interval $[0,1]$, the value one corresponding to integer $A$, zero corresponds to half-integer values of $A$; in the latter case spectral bands may appear only in the vicinity of the points $k=n\in\mathbb{N}$. A plot of probability with respect to $A$ is illustrated in Fig.~\ref{fig,prob,sym,ell1=0}.
\begin{figure}[!htb]
\centering
\includegraphics[scale=1.3]{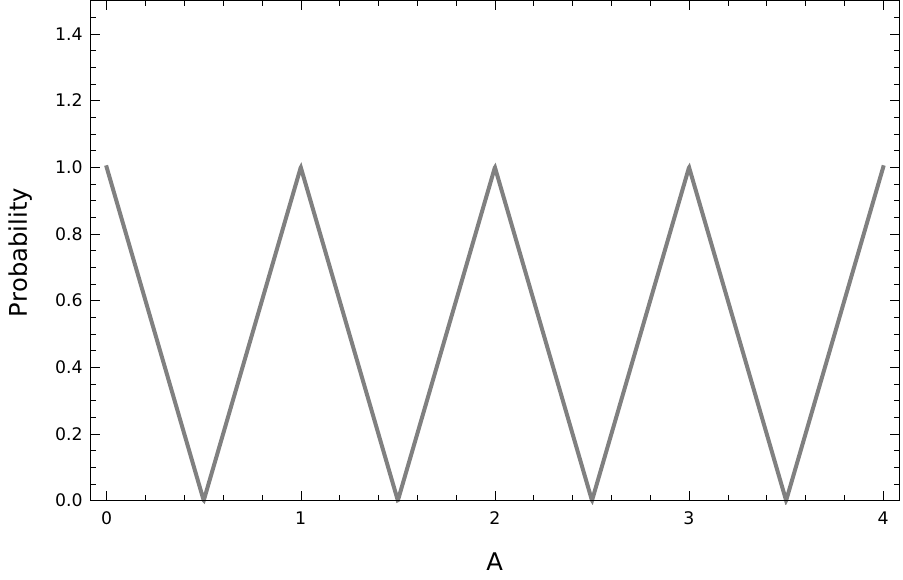}
\caption{The probability \eqref{prob,irrat,ell-1,sym} of the model $\ell_1=0$ versus the potential value $A$ for the symmetric chain $\ell_2=\ell_3=\pi$.}
\label{fig,prob,sym,ell1=0}
\end{figure}
Summarizing the obtained results for the probabilities of belonging to the (positive) spectrum, we have
\begin{equation*}
P_{\sigma}(H)= \left\{ \begin{array}{lcl} \dfrac12+2A-4A^2\quad (A \bmod \frac12) & \quad\dots\quad & \ell_3\neq\pi \,,\;\; \frac{\ell_2}{\ell_3}\notin\mathbb{Q} \\[.4em]
 1-\dfrac{1}{\pi}\,\arccos\left(\cos2A\pi\right) & \quad\dots\quad & \ell_3=\pi  \end{array} \right.
\end{equation*}
Comparing the two expressions we see the effect noted already in the non-magnetic case \cite{BET20, BET21}: in accordance with the result of \cite{BLS19} the model spectrum of the `loose-chain' Hamiltonian converges to that of the `tight' one as $\ell_1\to 0$, however, the limit is non-uniform with respect to the energy.

\subsection{Negative spectrum}
\label{sect:Mod-ell1=0,neg}
As for the `loose' chain, the negative spectrum corresponds to $k=i\kappa$ with $\kappa>0$; in this case, the functions $a$, $b$, and $c$ in the spectral condition \eqref{formul,abc} are as follows,
\begin{align}\label{Neg-ell-1=0,abc}
a=& \; (1-\kappa ^2 \ell ^2 ) \left(\cos A \ell_3 \, \sinh \kappa(2 \pi -\ell_3 ) +\cos A (2 \pi -\ell_3 ) \, \sinh \kappa\ell_3  \right)      \notag\\ \nonumber
& + \; 2 \kappa  \ell \left( \sin A (2 \pi -\ell_3 )\, \cosh  \kappa\ell_3  +  \sin A\ell_3 \,\cosh \kappa(2 \pi -\ell_3 )    \right) , \\[10pt] \nonumber
b=& \;  (\kappa ^2 \ell ^2-1) \left(\sin A (2 \pi -\ell_3 )\, \sinh \kappa\ell_3 -\sin A \ell_3 \, \sinh \kappa(2 \pi -\ell_3 ) \right)         \notag \\
& + \; 2 \kappa\ell \left( \cos A (2 \pi-\ell_3 )\, \cosh \kappa\ell_3 -  \cos A\ell_3 \, \cosh \kappa(2 \pi -\ell_3 ) \right),  \\[10pt]
c=& \;  2 \kappa  \ell  \sin 2A\pi + \left(1-\kappa ^2 \ell^2\right)\sinh 2\kappa\pi\notag
\end{align}
and the spectrum is given by the band condition \eqref{bandCon,a2b2c2} together with \eqref{Neg-ell-1=0,abc}. Similarly to the `loose' chain case, there is no flat band, however, here it holds in general only if $A\notin\mathbb{Z}$. To check the last claim, recall the symmetric non-magnetic case \cite{BET21} where we have a single isolated negative eigenvalue $-\ell^{-2}$ of infinite multiplicity. Here we have for $A\notin\mathbb{Z}$ in the symmetric situation a single continuous band, illustrated on Figs.~\ref{fig-symm-ell=2} and \ref{figell10}; the spectral condition in this case reads explicitly
\begin{equation*}\label{ell-1=0,SC-Neg,symm}
\left( 2  \left(1-\kappa^2 \ell ^2\right) \sinh \kappa\pi \, \cos A\pi +4 \kappa \ell \, \sin A\pi \, \cosh \kappa\pi    \right) \cos\theta = 2 \kappa \ell \, \sin 2A\pi+\left(1-\kappa^2 \ell ^2\right) \sinh 2\kappa\pi \,,
\end{equation*}
and accordingly, a number $-\kappa^2$ belongs to a spectral band provided
\begin{equation}\label{ell1-0,sym,neg}
\abs{\, 2  \left(1-\kappa^2 \ell ^2\right) \sinh \kappa\pi \, \cos A\pi +4 \kappa \ell \, \sin A\pi \, \cosh \kappa\pi } \,  \geq \; \abs{\, 2 \kappa \ell \, \sin 2A\pi+\left(1-\kappa^2 \ell ^2\right) \sinh 2\kappa\pi }.
\end{equation}
Note that since the elementary cell in this case contains a single vertex of degree four, the corresponding matrix $U$ has only one eigenvalue in the upper complex halfplane, and consequently, the negative spectrum cannot have more than a single band in accordance with Theorem 2.6 in \cite{BET21}. We are also able to localize it: for $\kappa=\ell^{-1}$ the band condition \eqref{ell1-0,sym,neg} reduces to $\left\lvert \cosh \frac{\pi}{\ell} \right\lvert \geq \left\lvert \cos A\pi\right\lvert$, implying that the energy $-\ell^{-2}$ always belong to the band. It is also clear from \eqref{ell1-0,sym,neg} that the band shrinks to the eigenvalue $-\ell^{-2}$ as $A$ approaches an integer value.

\section{The case $\ell_2=0$}
\label{sect:Mod-ell3=0}

Let us now consider the second degenerate situation when the connecting links are present but the two contacts at a ring merge to a single point so that $\ell_2=0$ and $\ell_3=2\pi$. As before, we will discuss the positive and negative spectrum separately.

\subsection{Positive spectrum}
\label{sect:Mod-ell3=0,pos}
In this case, $a$, $b$, and $c$ in the spectral condition \eqref{formul,abc} are as follows,
\begin{align}\label{Pos-ell3=0abc}
a=& \;   2 k \ell \,  \sin 2 A\pi +\left(k^2 \ell ^2+1\right) \sin 2k\pi     , \\[4pt] \nonumber
b=& \;  2 k \ell  \, \left(  \cos 2A\pi -\cos 2k\pi \right)    ,     \\[4pt]
c=& \;  2 k \ell \, \sin 2A\pi  \, \cos k \ell _1-\left(k^2 \ell ^2+1\right) \left(\cos 2A\pi \, \sin k\ell _1-\sin k \left(\ell _1+2 \pi \right)\right)     \nonumber  .
\end{align}
As before, we first consider the special case $A-\frac 12\in\mathbb{Z}$ where the spectral condition reduces to
\begin{equation}
\cos k\pi \left(  \left(k^2 \ell ^2+1\right) \sin k\pi \,\cos\theta   - 2\kappa\ell\cos k\pi \,   \sin\theta  - \left(k^2 \ell ^2+1\right) \sin k(\ell_1 +\pi )   \right)=0,
\end{equation}
indicating again that flat bands occur at $k=n-\frac 12$ with $n\in\mathbb{N}$ for all values of $\ell_1>0$. Away from these flat bands, the spectrum is absolutely continuous having a band-and-gap structure where the spectral bands satisfy the condition
$$   4 k^2 \ell ^2+\left(k^2 \ell ^2+1\right)^2 \cos 2k(\ell_1 +\pi )-\left(k^2 \ell
   ^2-1\right)^2 \cos 2k\pi \geq 0 .$$
In the generic situation, $2A\notin\mathbb{Z}$, the energy $\ell^{-2}$ again always belongs to the spectrum provided $A+\ell^{-1}\in\mathbb{Z}$. This can be easily checked by manipulating the equations $a=0$ and $b=0$ which yields $\frac{4 k^2 \ell ^2}{(k^2 \ell ^2+1)^2}=1$ as a necessary condition; it is fulfilled for $k=\ell^{-1}$. Inspecting then values of $a$, $b$, and $c$ at this point, we check easily that they all vanish for $A+\ell^{-1}\in\mathbb{Z}$. The rest of the spectrum is absolutely continuous having a band-gap structure; its points are given by the condition \eqref{bandCon,a2b2c2} together with \eqref{Pos-ell3=0abc}. The band-gap pattern for particular values of parameters is illustrated in Fig.~\ref{figell30}.
\begin{figure}[!htb]
\centering
\includegraphics[scale=1.1]{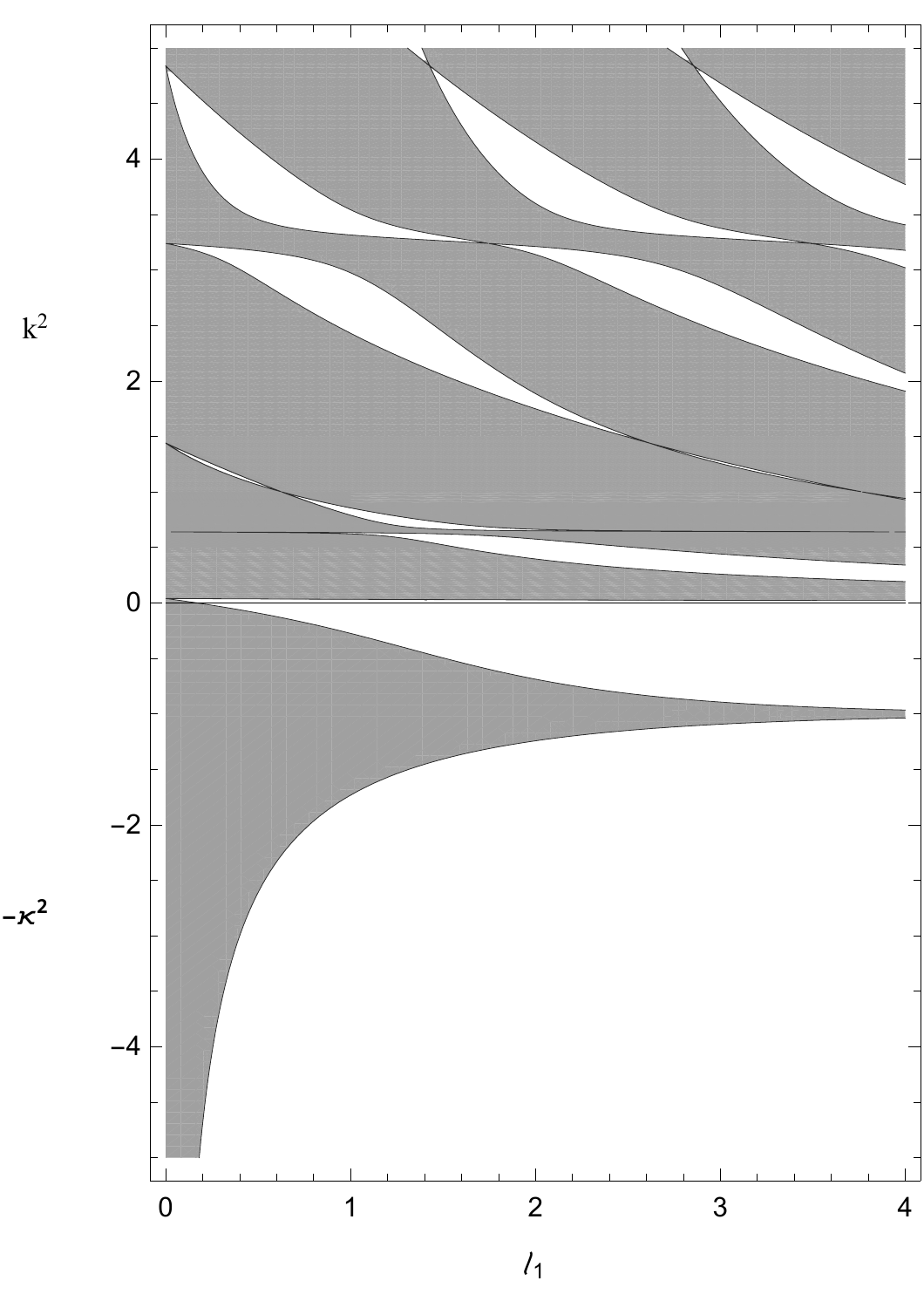}
\caption{ Spectrum of the model with $\ell_2=0$ in dependence on $\ell_1$ for $\ell=1$ and $A=\frac 15$ (the band condition \eqref{bandCon,a2b2c2} together with \eqref{Pos-ell3=0abc} and \eqref{Neg-ell-3=0,abc}, for the positive and negative part, respectively).}
\label{figell30}
\end{figure}
In particular, in the high-energy regime, $k\rightarrow\infty$, the band condition can be rewritten as
\begin{equation}\label{highCon,ell_3=0}
 64\, k^4 \ell^4  \left( \sin ^2 2k\pi-\left( \sin k(\ell_1 +2 \pi )-\cos 2A\pi \,\sin k\ell_1  \right)^2  \right)+\mathcal{O}(k^3) \geq 0.
\end{equation}
This determines the asymptotic behavior of the bands and, in particular, the probability \eqref{probsigma} that a randomly chosen energy belongs to the spectrum. If $\ell_1$ is a rational multiple of $\pi$, the function in the large brackets in \eqref{highCon,ell_3=0} is periodic, and accordingly, one can compute the probability numerically by finding the roots of the function determining the fraction of the period where the function is non-negative.
It is difficult to find the probability \eqref{probsigma} in a closed form, but it depends both on $\ell_1$ and $A$ as can be seen in Figs.~\ref{fig-prob-ell-3=0-A=1-5om} and \ref{fig-probe-ell-3=0-ell1=pi-5om}.
\begin{figure}[!htb]
\centering
\includegraphics[scale=1]{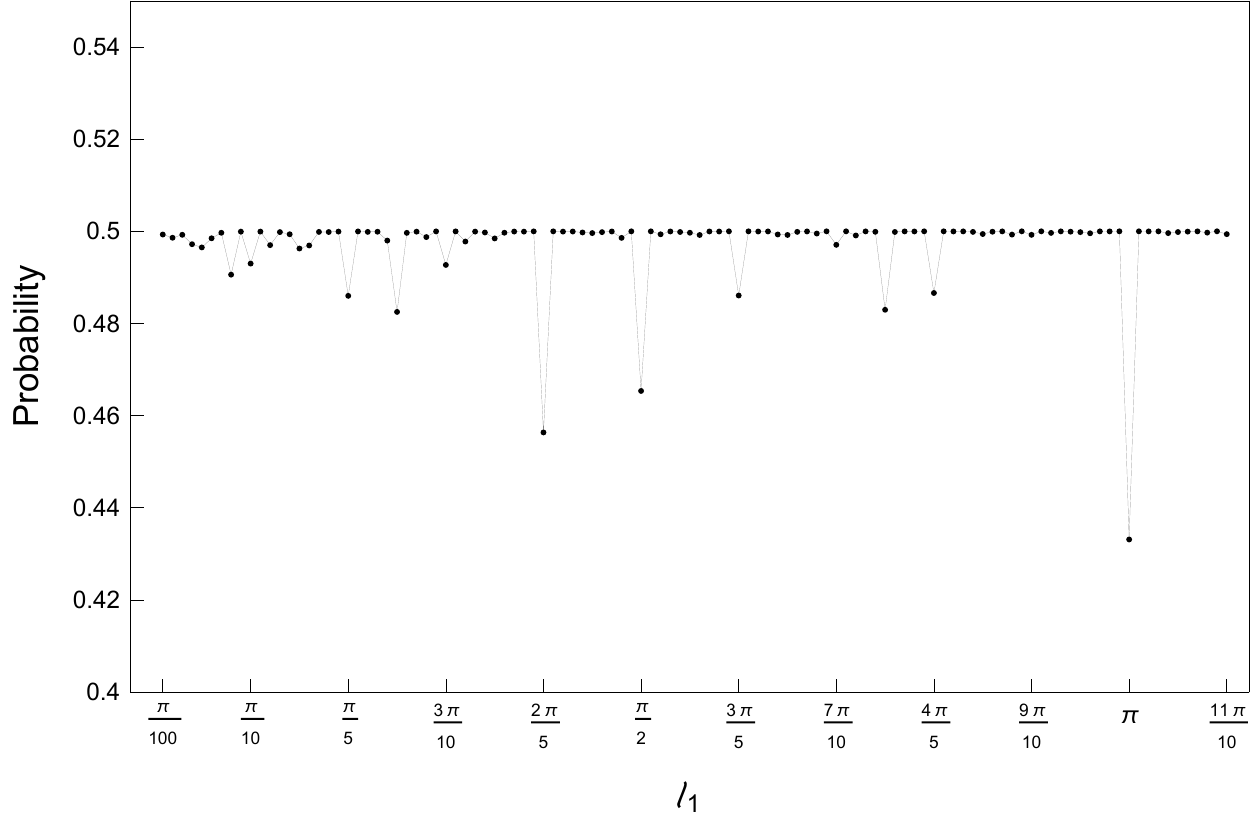}
\caption{The probability \eqref{probsigma} of the model with $\ell_2=0$, obtained by the condition \eqref{highCon,ell_3=0}, \emph{vs.} the edge length $\ell_1$ being a rational multiple of $\pi$ and $A=\frac 15$. }
\label{fig-prob-ell-3=0-A=1-5om}
\end{figure}
\begin{figure}[!htb]
\centering
\includegraphics[scale=1]{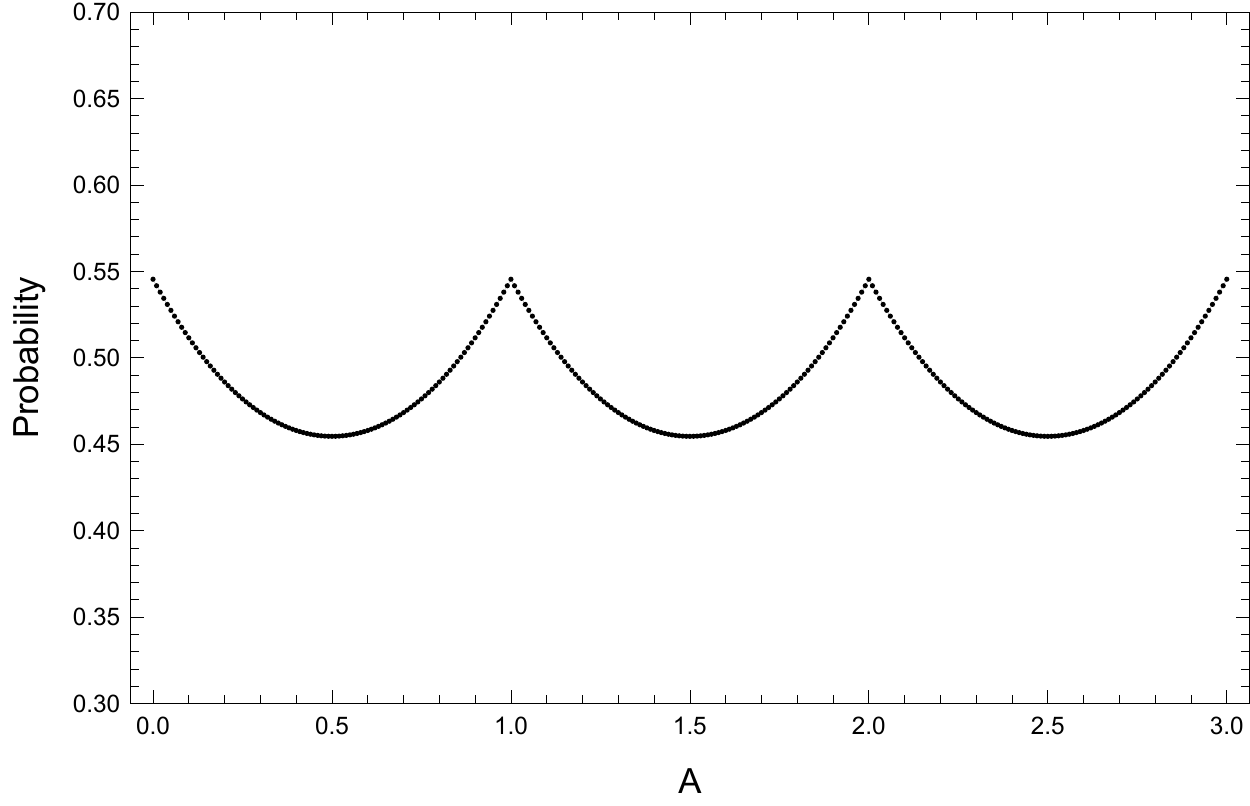}
\caption{The probability \eqref{probsigma} of the model with $\ell_2=0$, obtained by the condition \eqref{highCon,ell_3=0}, \emph{vs.} the potential value $A\in[0,3]$ for $\ell_1=\frac{\pi}{5}$. }
\label{fig-probe-ell-3=0-ell1=pi-5om}
\end{figure}
On the other hand, if $\ell_1$ and $\ell_3=2\pi$ are incommensurate, the probability \eqref{probsigma} equals $\frac 12$ as in the non-magnetic case \cite{BET21} which means that the Band-Berkolaiko universality holds again, and moreover, that the probability is independent of $A$. To see this, note that with $k$ running over positive numbers, one can regard $x:=k\ell_1\,(\bmod\, 2\pi)$ and $y:=k\pi\,(\bmod\, 2\pi)$ as a pair of independent identically distributed random variables on $[0,2\pi)$, and consequently, the sought probability is the fraction of the total area $4\pi^2$ in which
\begin{equation}\label{area-ell3=0}
\sin^2(2y)-\left( \sin(x+2y)-\cos2A\pi \,\sin x \right)^2 \geq 0
\end{equation}
holds. This non-negative region is again a union of irregular geometric shapes, cf. Fig.~\ref{figProb-ell30-3D};
\begin{figure}[!htb]
\centering
\includegraphics[scale=0.75]{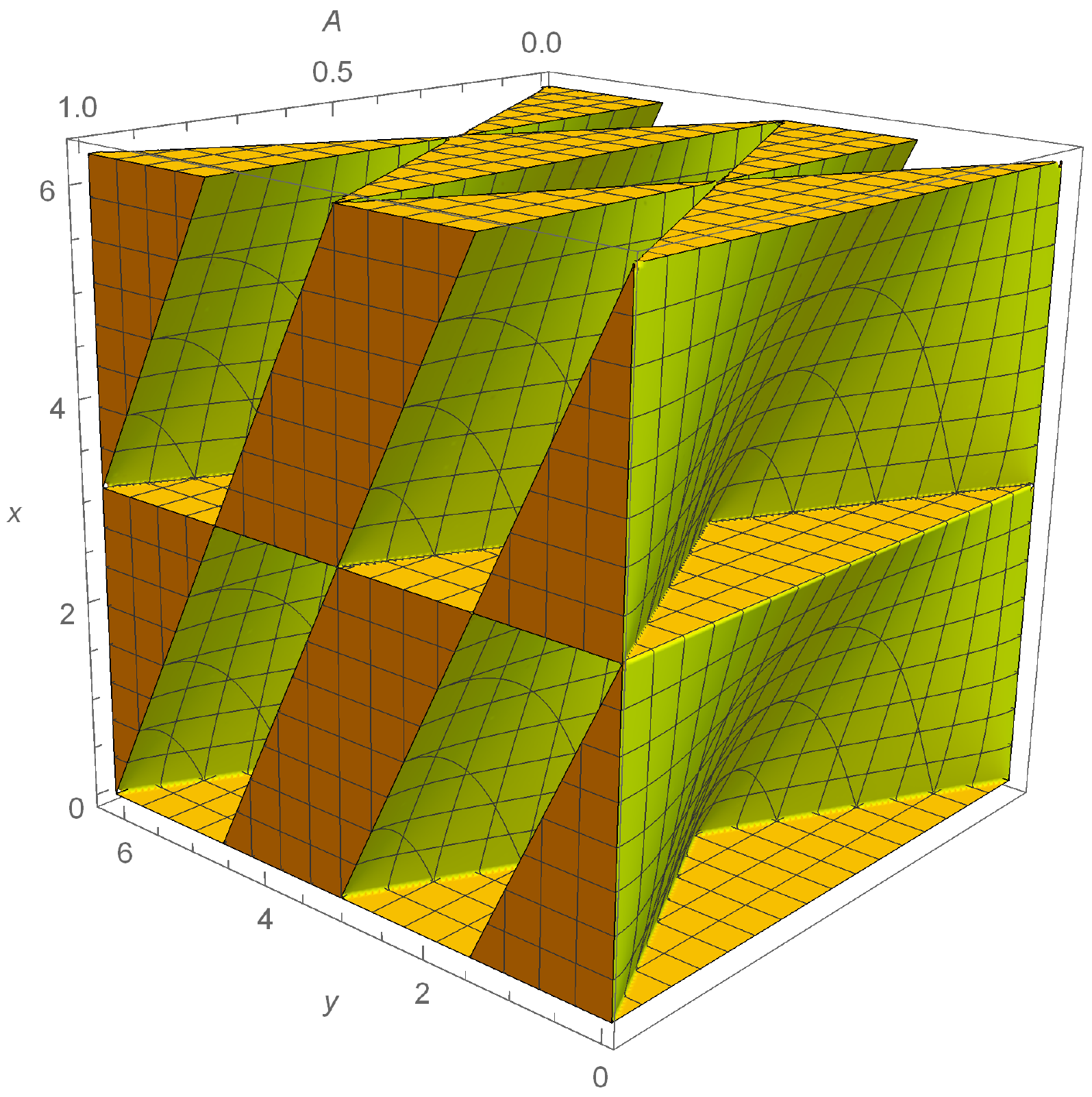}
\caption{The three-dimensional region in which the condition \eqref{area-ell3=0} holds.}
\label{figProb-ell30-3D}
\end{figure}
one can easily check that the function on the left-hand side of \eqref{area-ell3=0} is symmetric with respect to $x\leftrightarrow x-2\pi$ and $y\leftrightarrow y-\pi$, thus it again suffices to investigate the first octant only, i.e. the intervals $x\in[0,\pi)$ and $y\in[0,\frac{\pi}2)$. The inequality \eqref{area-ell3=0} can be rewritten in the form
\begin{equation}\label{area-ell3=0-abs}
\left\lvert \sin(x+2y)-\mathcal{A} \,\sin x \right\lvert \leq  \left\lvert\sin2y\right\lvert
\end{equation}
with $\mathcal{A}:=\cos2A\pi$. Taking then into account that all the trigonometric functions $\sin 2y$, $\sin \frac x2$ and $\cos \frac x2$ are positive over the specific intervals, one can check that \eqref{area-ell3=0-abs} is equivalent to the inequalities
$$ \mathcal{A}\,\sin \frac x2\,\leq \sin \Big(\frac x2 +2y\Big) \quad \text{and}\quad \mathcal{A}\,\cos \frac x2\,\geq \cos \Big(\frac x2 +2y\Big) .$$
This means that the sought area is enclosed between the curves
\begin{align*}
  y_2(x,\mathcal{A}) &= \frac{\pi}4-\frac x4+\frac12\arccos\Big(\mathcal{A}\,\sin \frac x2\Big)  , \\[5pt]
  y_1(x,\mathcal{A}) &=  -\frac x4+\frac12\,\arccos\Big(\mathcal{A}\,\cos \frac x2\Big).
\end{align*}
Then it is easy to check that
\begin{align*}
 \frac{d}{d\mathcal{A}}\int_{0}^{\pi} \left(y_2(x,\mathcal{A})-y_1(x,\mathcal{A}) \right)\,dx &=-\frac12 \int_{0}^{\pi} \frac{\sin\frac x2}{\sqrt{1-\mathcal{A}^2\,\sin^2\frac x2}} \,dx +\frac12\int_{0}^{\pi} \frac{\cos\frac x2}{\sqrt{1-\mathcal{A}^2\,\cos^2\frac x2}} \,dx \\
&= -\frac1{\mathcal{A}}\,\arctanh\mathcal{A}+\frac1{\mathcal{A}}\,\arctanh\mathcal{A}=0.
\end{align*}
This proves our claim that the value of the magnetic potential $A$ does not affect the probability and it suffices to compute the integral for $A=0$ which gives $\frac{\pi^2}{4}$; multiplying then the result by eight, and dividing by $4\pi^2$, one gets $\frac12$ as in \cite{BET21}. Note also that approximating such an irrational edge lengths relation by rationals, we approach the $\frac12$ as Fig.~\ref{fig-prob-ell-3=0-A=1-5om} shows.

\subsection{Negative spectrum}
\label{sect:Mod-ell3=0,neg}
To find the negative part of the spectrum, one has to substitute $k=i \kappa$ into \eqref{Pos-ell3=0abc} which gives
\begin{equation}
\begin{aligned}\label{Neg-ell-3=0,abc}
a=& \;    2 \kappa\ell \, \sin 2A\pi+ \left(1-\kappa ^2 \ell^2\right) \sinh 2 \kappa \pi    , \\[5pt]
b=& \;     2 \kappa\ell  \left(   \cos 2A\pi -\cosh 2 \kappa\pi    \right)  ,  \\[5pt]
c=& \;    \left(\kappa ^2 \ell ^2-1\right) \left(  \cos 2A\pi \, \sinh \kappa\ell_1 -\sinh\kappa(\ell_1 +2 \pi )     \right)+2 \kappa \ell  \, \sin 2A\pi  \cosh \kappa\ell_1   .
\end{aligned}
\end{equation}
There is again no flat band; the spectrum is given by the band condition \eqref{bandCon,a2b2c2} together with \eqref{Neg-ell-3=0,abc}. As in the previous case, the elementary cell contains a single vertex of degree four, hence the spectrum cannot have more than a single negative band in accordance with Theorem 2.6 of \cite{BET21}. Furthermore, mimicking the argument used for the general model, we conclude this band shrinks to a point in the limit $\ell_1\rightarrow\infty$. In that case, the function $f(\ell,\ell_3,A;\kappa)$ in \eqref{large,ell-1,gen} reads
\begin{equation}\label{large,f,ell3-0}
f(\ell,2\pi,A;\kappa):=4  (\kappa ^2 \ell ^2-1 ) \big( \cos 2A\pi-e^{2\kappa\pi} \big)+8 \kappa\ell  \sin 2A\pi,
\end{equation}
which allows one to check easily that for $A\in\mathbb{Z}$ and $A-\frac12\in\mathbb{Z}$, the band shrinks to the energy $-\ell^{-2}$. The former case has been investigated in \cite{BET21}; to estimate the width of the shrinking band for $A-\frac12\in\mathbb{Z}$, we set $\kappa=\ell^{-1}+\delta$; substituting into the spectral condition and solving the resulting equation for $\delta$, we obtain the following asymptotic expressions for the energy and width of the band
\begin{align*}
 -\kappa^{2}&= -\frac{1}{\ell^{2}}-\frac{2}{\ell^{2}}\,\left(1+\e^{-\frac{2\pi}{\ell}}\right)\sin\theta \; \e^{-\frac{\ell_1}{\ell }}+\mathcal{O}(\e^{-2\frac{\ell_1}{\ell }})    ,   \\[7pt]
 \Delta E &= \frac{4}{\ell^{2}}\,\left(1+\e^{-\frac{2\pi}{\ell}}\right) \,\e^{-\frac{\ell_1}{\ell }}+\mathcal{O}(\e^{-2\frac{ \ell_1}{\ell }}) .
\end{align*}
For other values of $A$, the shrinking is still exponential, however, the limiting point, determined by the condition $f(\ell,2\pi,A;\kappa)=0$, is in general different from $\ell^{-2}$.

\subsection*{Acknowledgements}

P.E. was supported by the Czech Science Foundation within the project 21-07129S and by the EU project CZ.02.1.01/0.0/0.0/16\textunderscore 019/0000778. M.B. and J.L. were supported by the Czech Science Foundation within the project 22-18739S. M.B.’s work was also supported by the Internal Postdoc Project UHK for years 2021-2022.


\begin{thebibliography}{99}

\bibitem{BB13}
R.~Band, G.~Berkolaiko, Universality of the momentum band density of periodic networks, \emph{Phys. Rev. Lett.} \textbf{113} (2013), 13040.

\bibitem{BE21}
M.~Baradaran, P.~Exner: Kagome network with vertex coupling of a preferred orientation, \texttt{arXiv:2106.16019}.

\bibitem{BET20}
M.~Baradaran, P.~Exner, M.~Tater, Ring chains with vertex coupling of a preferred orientation, \emph{Rev. Math. Phys.} \textbf{33} (2021), 2060005.

\bibitem{BET21}
M.~Baradaran, P.~Exner, M.~Tater: Spectrum of periodic chain graphs with time-reversal non-invariant vertex coupling, \texttt{arXiv:2012.14344}.

\bibitem{BK13}
G.~Berkolaiko, P.~Kuchment: \emph{Introduction to Quantum Graphs}, AMS, Providence, R.I., 2013.

\bibitem{BLS19}
G.~Berkolaiko, Y.~Latushkin, S.~Sukhtaiev, Limits of quantum graph operators with shrinking edges, \emph{Adv. Mat.} \textbf{352} (2019), 632--669.

\bibitem{EKKST}
P.~Exner, J.P.~Keating, P.~Kuchment, T.~Sunada, A.~Teplayaev, eds.: \emph{Analysis on graphs and its applications}, Proc. Symp. Pure Math., vol. 77; Amer. Math. Soc., Providence, R.I., 2008.

\bibitem{EL19}
P.~Exner, J.~Lipovsk\'{y}: Spectral asymptotics of the Laplacian on Platonic solids graphs, \emph{J. Math. Phys.} {\bf 60} (2019), 122101.

\bibitem{EM17}
P.~Exner, S.~Manko: Spectral properties of magnetic chain graphs, \emph{Ann. H. Poincar\'{e}} {\bf 18} (2017), 929--953.

\bibitem{EP13}
P.~Exner, O.~Post: A general approximation of quantum graph vertex couplings by scaled Schr\"odinger operators on thin branched manifolds, \emph{Commun. Math. Phys.} \textbf{322} (2013), 207--227.

\bibitem{ET18}
P.~Exner, M.~Tater: Quantum graphs with vertices of a preferred orientation, \emph{Phys. Lett.} \textbf{A382} (2018), 283--287.

\bibitem{ET21}
P.~Exner, M.~Tater: Quantum graphs: self-adjoint, and yet exhibiting a nontrivial $\mathcal{PT}$-symmetry, \emph{Phys. Lett.} {\bf A416} (2021), 127669.

\bibitem{Ja56}
V.~Jarn\'{\i}k: \emph{Differential calculus II}, Academy Publishers, Prague 1956.

\bibitem{KS03}
V.~Kostrykin, R.~Schrader: Quantum wires with magnetic fluxes, \emph{Commun. Math. Phys.} \textbf{237} (2003), 161--179.

\bibitem{SK15}
P.~St\v{r}eda, J.~Ku\v{c}era, Orbital momentum and topological phase transformation, \emph{Phys. Rev.} \textbf{B92} (2015), 235152.

\end{thebibliography}
\end{document}